 \documentclass[12pt]{article}
\usepackage{amsfonts}

\usepackage{graphicx}
\usepackage{amsmath}
\usepackage{hyperref}
\usepackage{color}

\newtheorem{theorem}{Theorem}[section]

\newtheorem{corollary}[theorem]{Corollary}

\newtheorem{definition}[theorem]{Definition}
\newtheorem{example}[theorem]{Example}

\newtheorem{lemma}[theorem]{Lemma}

\newtheorem{notation}[theorem]{Notation}

\newtheorem{remark}[theorem]{Remark}

\newtheorem{terminology}[theorem]{Terminology}
\newenvironment{proof}[1][Proof]{\textbf{#1.} }{\ \rule{0.5em}{0.5em}}

\def \L{\Lambda}
\def \<{\langle}
\def \>{\rangle}

\def \as{\mathsf{a}}
\def \bs{\mathsf{b}}

\def \A{{\cal A}}

\def \Bs{\mathsf{B}}

\def \C{\mathbb C}

\def \Es{\mathsf{E}}

\def \R{\mathbb R}
\def \Ra{\mathcal R}

\def \G{{\cal G}}

\def \P{{\cal P}}

\def \cg{{\cal C}}

\def \a{{\bf a}}

\def \w{\omega}
\def \k{{\bf k}}
\def \kf{\frak k}

\def \b{{\bf b}}

\def \ua{\mathsf{u}}

\def \va{\mathsf{v}}

\def \xb{{\bf x}}

\def \p{\partial}
\def \beq{\begin{equation}}
\def \eeq{\end{equation}}

\def \n{\nabla}
\def \eref{\eqref}

\def \s{\sigma}
\def \lrc{\lrcorner\,}
\def \({\Big(}
\def \){\Big)}
\numberwithin{equation}{section}

\def \Ws{\mathsf{W}}

\begin{document}

\title{Equivalence of helicity and Euclidean self-duality for gauge fields
\footnote{\emph{Key words and phrases.} 
Helicity, self duality, Yang-Mills, non-abelian gauge theories \newline
 \indent 
\emph{2010 Mathematics Subject Classification.} 
 Primary; 81T13, 78A25,  Secondary; 78A40.
 }    }     
\author{Leonard Gross \\
Department of Mathematics\\
Cornell University\\
Ithaca, NY 14853-4201\\
{\tt gross@math.cornell.edu}}

\maketitle

\begin{abstract}  
In the canonical formalism for the free electromagnetic field  a solution to Maxwell's
equations is customarily identified with  its initial gauge potential (in Coulomb gauge) and
initial electric field, which together determine a point in phase space. The  solutions
to Maxwell's equations, all of whose plane waves in their plane wave expansions have
positive helicity, thereby determine a subspace of phase  space. We will show that this
subspace consists of initial gauge potentials which lie in the positive spectral subspace
of the operator curl together with initial electric fields conjugate to such  potentials.
Similarly for negative helicity. Helicity is thereby characterized by the spectral subspaces
of curl in configuration space. A gauge potential on three-space has 
a  Poisson extension to a four dimensional Euclidean half space, defined as the solution
to the Maxwell-Poisson equation whose initial data is the given gauge potential.
We will  show that the extension is anti-self dual if and only if the gauge potential lies
in the positive spectral subspace of curl. Similarly for self dual extension and  negative
spectral subspace.  Helicity is thereby characterized  for a normalizable electromagnetic
field by the canonical formalism and (anti-)self duality.

For a non-abelian gauge field on Minkowski space a plane wave expansion is not
gauge invariant. Nor is the notion of positive spectral subspace of curl.          
But if one replaces the  Maxwell-Poisson equation by the Yang-Mills-Poisson equation
then  (anti-)self duality on the Euclidean side induces a decomposition of (now non-linear)
configuration space similar to that in the electromagnetic case.   
The strong analogy  suggests  a gauge invariant definition of helicity for
non-abelian gauge fields. We will provide further support for this view. 
 \end{abstract}

\tableofcontents

\section{Introduction}     \label{secint}

A workable notion of helicity for quantized Yang-Mills fields has been  sought in many works in recent years.
For a broad perspective on this work see  the extensive review 
 by Leader and Lorc\'e   \cite{LL2014} .     
       Any notion of helicity for Yang-Mills fields must  be gauge invariant, Lorentz invariant and
 reduce to the standard notion of helicity for electromagnetic fields. 
        This paper proposes a notion of helicity for non-abelian gauge fields based on a  
  Euclidean   interpretation  of helicity in electromagnetism.
  For the latter, it will first be shown that the Poisson extension of the initial gauge potential 
  of an electromagnetic field
  to a half space  in Euclidean space-time sets up an equivalence  between helicity, 
  which is a Minkowski space notion, and (anti-)self duality,
 which is a Euclidean space notion.
 
  For classical electromagnetic fields helicity is customarily defined in terms of plane waves,   
  which do not mesh    
  well with a non-linear theory such as Yang-Mills, \cite{Act1}.  
  To make   the jump from electromagnetism to Yang-Mills fields, first classically (in this paper) 
  and then quantum mechanically (in \cite{G76}),
  we are going to describe helicity in the electromagnetic case in a manner that allows
   a natural    extension
  to the Yang-Mills  case.  
  To begin, we will first describe the precise configuration space for the free
   electromagnetic field that reflects 
  the unique Lorentz invariant norm of Bargmann and Wigner, \cite{BW}. 
  Second, given a real valued gauge 
  potential  $A(x)$ on $\R^3$ we will show that its Poisson 
  extension to a half
  space in $\R^4$ not only captures  the Lorentz invariant norm of Bargmann and Wigner, but is also 
  self-dual
  if and only if $A$ is the initial data in Minkowski space of some solution to Maxwell's
   equations of negative helicity. Similarly for
  anti-self dual and positive helicity.
 
              To be more explicit, 
  suppose that   $A(x):=\sum_{j=1}^3 A_j(x) dx^j$  is a real valued gauge potential on $\R^3$
  and that, for  each $s \ge 0$,    $\as(x,s) = \sum_{j=1}^3 \as_j(x,s) dx^j$  is another
 gauge potential on $\R^3$  
 such that the function   
 $\as$ satisfies the Maxwell-Poisson equation 
 with initial value $A$:
 \begin{align}
 \as''(s)  = curl^2\ \as(s),\ \ \ \  \as(0) = A.    \label{IMP1}
 \end{align}
  We may  regard $\as$ as a 1-form  in temporal gauge on the half-space
    $\R^4_+ := \R^3 \times [0,\infty)$.     Its four dimensional
        curvature $F$ is  then given by $F = ds\wedge \as' + d\as$,
         where $d$ denotes the three dimensional exterior derivative. Equation \eref{IMP1} is the Euler
         equation for minimization of $\int_{\R^4_+} |F(x,s)|^2 d^3x\, ds$,  subject to the 
         initial condition $\as(0) = A$.       
 It happens that this minimum  is the square of the unique Lorentz compatible 
  norm on configuration space. 
We will show  that   $A$ is the initial potential of some
solution to Maxwell's equations with only 
 positive helicity plane waves in its plane wave decomposition
if and only if  $F$ is anti-self-dual. Similarly for negative helicity  and self-duality.  
 The Maxwell-Poisson equation \eref{IMP1} thereby sets up a  
 correspondence
  between helicity on the 
 Minkowski  side and (anti-)self-duality on the Euclidean side. The result is a decomposition of the 
 electromagnetic configuration space $\cg$ into two orthogonal subspaces $\cg_\pm$, 
 determined by helicity, or  equivalently, by  (anti-)self duality.

  For a non-abelian gauge field the corresponding Poisson-like equation is the Yang-Mills-Poisson equation.
  It  is a nonlinear, degenerate elliptic differential equation    
  for which existence and uniqueness of solutions has not yet been proven. We are going to assume both,
  however, 
 and show how the Yang-Mills-Poisson equation  sets up a similar decomposition
 of the Yang-Mills configuration space $\cg$, which is no longer a linear space, 
  into two submanifolds $\cg_\pm$ corresponding 
  to anti-self dual and self dual solutions.  
 The similarity of this procedure to the electromagnetic case 
 suggests that the submanifolds $\cg_\pm$
 are the ``correct'' non-abelian analogs of the classical helicity subspaces for Maxwell's theory. 
But we will give more 
support to this interpretation in 
Section \ref{secdecomp}  by  showing that this decomposition   exhibits two additional properties
which are definitive in the electromagnetic case.

The representation of a gauge potential as the initial data of a Poisson-like equation automatically
induces a Riemannian metric on the set of potentials (modulo gauge transformations), whose Poisson
 action $\int_{\R^4_+} |F(x,s)|^2 d^3x ds$ is finite. This in turn gives a quantitative meaning
  to phase space $T^*(\cg)$  for both  electromagnetic and Yang-Mills fields. 
   It turns out, remarkably,
   that the  norm defined  on the electromagnetic  phase space in this way is exactly the unique
   Lorentz invariant norm first discovered by Bargmann and Wigner \cite{BW} in their fundamental work
   on unitary representations of the  inhomogeneous Lorentz group. 
  The norm induced on configuration space and phase space by our procedure 
  is therefore   forced on us if we wish to have a Lorentz invariant theory.      
   We will show this in the course
   of establishing the equivalence of helicity with (anti-)self duality for electromagnetism.
   In the non-abelian case
   the Riemannian metric induced from the Euclidean side by the same procedure is presumably
   Lorentz invariant in some appropriate sense. But this has not been explored.
   
   In accordance with the canonical formalism (as in, say,  \cite[Chapter 14]{BD}) the electric
   field $E$ at time zero is canonically conjugate  to $A$ and therefore 
   the pair $\{A,E\}$   
   defines a point in phase space. 
   We will show that if $A$ is in $\cg_+$ and $E$ is in $T_A^*(\cg_+)$
   then the unique solution to  Maxwell's equations with initial data $A,E$  contains only positive 
   helicity plane waves in its plane wave decomposition (and conversely). 
   Similarly $T^*(\cg_-)$ corresponds to negative    
   helicity plane waves. 
   In this way  Euclidean (anti-)self duality provides  a completely geometric characterization of helicity
   in electromagnetic theory

   We will  reinforce this geometric interpretation   of helicity  
   further  by starting over again  in \cite{G76} 
   with the quantized theories, each 
   of which    has      
 a generally accepted  ``forward component of angular momentum'' operator 
  used for  defining
  helicity in their respective quantum field theories.   
  In the electromagnetic  case we will see that  the operator simply ``quantizes'' the decomposition $\cg_\pm$. 
 The Schr\"odinger representation of the quantized electromagnetic field will be instrumental
  in implementing this view.
            In the non-abelian  case,  where configuration space is not a linear space, 
 we will   see that  the apparent   ``forward component of angular momentum'' operator similarly quantizes 
 the two submanifolds $\cg_\pm$,  thereby reinforcing   
again   the link between helicity and Euclidean (anti-)self-duality for non-abelian gauge fields. 
 The Schr\"odinger representation of the quantized Yang-Mills field will again be indispensable 
 in implementing this view, and for this, the Riemannian metric introduced in this paper will
 enable us to use the gradient operator on functions over configuration space instead
  of momentum space annihilation operators, which have no gauge invariant meaning
   for the intermediate  non-abelian fields.

  There are many unsettled  purely mathematical issues  in the body of this paper
   that we will ignore. They are in the nature of existence and uniqueness theorems.
   In the last section we will  discuss   
  the open problems raised by these structures.

\section{Plane waves, helicity, configuration space in electromagnetism}  \label{secpw}
\subsection{Review of helicity for plane waves}     \label{sechpw}

This subsection is a review of  
 circularly polarized light in a form that will be useful for our purposes.
 Nothing in this subsection should be construed  to be less than 165 years old.
 See for example Born and Wolfe \cite{BoW}  or Jackson \cite{Jac}. 
 
\begin{notation} {\rm For $0 \ne \k \in \R^3$ define
\beq
\k^\perp = \{ u \in \C^3: u\cdot \k =0 \}.                      \label{hel33}
\eeq
 $\k^\perp$ is a two dimensional complex vector space closed under complex conjugation $u\mapsto \bar u.$
  Its real part
$\k^\perp_{real}:=\k^\perp \cap \R^3$ is a two dimensional real vector space.
Define
\beq
C_\k u = i\k\times u,\ \ \text{for}\ u \in \C^3.                
 \label{hel35}
\eeq
$C_\k u$ is in $\k^\perp$  for all $u \in \C^3$  because $\k \times u$ is perpendicular to $\k$. 
 The significance of the operator  $C_\k$ is that it implements  curl 
 under Fourier transformation  in accordance with the easily verified  identity
 \beq
 \text{curl} (u e^{i\k\cdot x}) =(C_\k u) e^{i\k\cdot x}\ \ \ 
                                \text{for any vector} \ u \in \C^3 .                \label{hel36}
  \eeq 
}
\end{notation}

Discussions of circularly polarized light and, more generally, elliptically polarized
light  amount to  an analysis of the operator $C_\k$ acting in the two 
dimensional subspace $\k^\perp$. 
See for example \cite[Chapter 7]{Jac} or \cite[Section 1.4]{BoW}.

\begin{lemma}\label{lemC}{\rm (Properties of $C_\k$).} The restriction of $C_\k$ to $\k^\perp$ satisfies
\begin{align}
&\ \ \ a)\ \ \text{$C_\k$ is Hermitian in $\k^\perp$.} \notag\\
&\ \ \ b)\ \  C_\k^2 = |\k|^2 on\ \k^\perp \label{hel37}  \\
&\ \ \ c)\ \ \text{$C_\k$ decomposes  $\k^\perp$ into two one dimensional complex subspaces
       $\k^\perp_{\pm}$}                                   \notag\\
&      \text{  corresponding to eigenvalues $\pm |\k|$. 
 A vector $u \in \k^\perp$ is in $\k^\perp_+$ if} \notag \\
 &\text{ and only if its complex conjugate $\bar u \in \k^\perp_-$.} \notag
 \end{align}
 \end{lemma}
          \begin{proof}
Using the triple product identity $(\k\times u)\cdot v = - u\cdot (\k\times v)$, which
is valid for $u$ and $v$ in $\C^3$, we see that
$$
(C_\k u, v)_{\C^3}= (i\k \times u) \cdot \bar v = u\cdot (\overline{i\k\times v}) =(u, C_\k v)_{\C^3},
$$
which proves a).  Item b) follows from the identity  $C_\k^2 u = -\k\times (\k \times u) = |\k|^2 u$
when $u\cdot \k =0$.  
To prove c)  let $\hat \k  = \k/|\k|$ and let $e_1$ be any unit vector in $\k^\perp_{real}$. 
Define $e_2 = \hat\k \times e_1$. Then $\hat\k, e_1, e_2$ form a right 
handed basis of $\R^3$. In particular $\hat \k \times e_2 = - e_1$. Thus 
\beq
C_\k (e_1+ie_2) = i|\k|\{e_2 -ie_1\} =|\k|(e_1 + i e_2)
\eeq
Hence $(e_1 + i e_2)$ is an eigenvector for $C_\k$ with eigenvalue $|\k|$. Using the same vectors one can compute similarly that $e_1 - ie_2$ is an eigenvector of $C_\k$
with eigenvalue $- |\k|$. Moreover, if $u \in \k^\perp_+$ then 
$C_\k \bar u = \overline {(-C_\k u)}  = \overline{-|\k| u} = -|\k| \bar u$. This proves c). 
  \end{proof}
  
  \begin{definition}{\rm   (Plane waves) Let $0 \ne \k \in \R^3$. A {\it plane wave}  
with wave vector $\k$  is a function on $\R^4$ of the form
\beq
A_\k(x,t) = a e^{i(\k\cdot x - |\k|t)} +  \bar{a} e^{-i(\k\cdot x - |\k|t)}, \   x,t\in \R^{3+1},\   a \in \k^\perp 
  \label{pw6}
\eeq
Clearly  $A_\k(x,t) \in \R^3$ for each $x, t$ and 
\beq
div\, A_\k(x,t) = 0 \label{pw6a} 
\eeq
 because of the identity 
\begin{align}
div\, (ae^{i\k\cdot x}) = (i\k\cdot  a)e^{i\k\cdot x}.    \label{pw6b}
\end{align}
}
\end{definition}
  
  \begin{lemma} \label{lempw} For the plane wave \eref{pw6} define 
\begin{align}
B_\k (x,t) &= curl A_\k(x,t) \ \       \label{pw7}\\
E_\k(x,t)  &= - (\p/\p t) A_\k(x,t).        \label{pw8}
\end{align}
Then $B_\k, E_\k$ is a solution to Maxwell's equations, 
\begin{align}
\n\cdot B=0,\ \ \n\cdot E = 0,\ \   \dot B =-curl\ E, \ \ \dot E = curl\ B.   \label{me}
\end{align}
\end{lemma}
      \begin{proof} Use the identities  \eref{pw6b} and  \eref{hel36}  along with \eref{pw6} and \eref{hel37}.
\end{proof}

\begin{definition}\label{defHP} {\rm (Helicity for plane waves) 
The standard definition of helicity can be stated thus. The plane
wave \eref{pw6} has {\it positive helicity} if $C_\k a = |\k| a$. It has {\it negative helicity} if $C_\k a = -|\k|a$.
}
\end{definition}

\subsection{The Lorentz invariant norm}            \label{secLinv}

A small portion of the classic paper   \cite{BW} of Bargmann and Wigner will
be needed to explain  our choice of 
the norm \eref{n2}, which we will later take to be the norm on configuration space in \eref{he201}. 
The Coulomb   gauge fixing that we are going to use here for convenience 
will be removed  when we discuss Yang-Mills fields.

\begin{definition} {\rm (Spaces)  
The operator curl on vector fields over $\R^3$ acts in the real Hilbert space  
 $L^2(\R^3; \R^3) \cap \{div\, u =0\}$ because $div\, curl =0$. It is easily seen to be self-adjoint by an integration
 by parts. Moreover in this space it has a zero nullspace. We will denote by $C$ the operator $curl$
 acting in this Hilbert space or some closely related spaces. 
 The square root of $C^2$ will be denoted  $|C|$: Thus 
 \begin{align}
 C = curl, \ \ \ \ |C| =\sqrt{C^2}         .                      \label{n1}
 \end{align}
 For a vector field $u$ on $\R^3$  with $div\, u =0$ let 
\begin{align}
\|u\|_{H_{1/2}}^2 &=  \|\ |C|^{1/2} u\|_{L^2}^2                 \label{n2} \\
\|u\|_{H_{-1/2}}^2 &=  \|\ |C|^{-1/2} u\|_{L^2}^2.               \label{n3}
\end{align} 
Define Hilbert spaces  $H_{\pm 1/2}$ by the condition that the corresponding
 norm in \eref{n2} or \eref{n3} is finite. 

For vector fields $B$ and $E$ on $\R^3$  with divergence zero define
\begin{align}
\|\{B,E\}\|_{bw}^2 = \| B\|_{H_{-1/2}}^2 +   \| E\|_{H_{-1/2}}^2   .       \label{n6}
\end{align}
}
\end{definition}
The main assertion of this section is that  the norm $\|\cdot, \cdot\|_{bw}$ is invariant under Lorentz transformations in the sense that the unique solution to Maxwell's equations with $B$ and $E$ as initial
data has, after any Lorentz transformation, initial data with the same norm. The following computations 
reduce this assertion to a theorem in \cite{BW} by showing that  $\|\cdot, \cdot\|_{bw}$ is equal
 to the norm introduced in \cite{BW} by Bargmann and Wigner.

       \begin{theorem}\label{thmpwd} 
$($Plane wave decomposition$)$ 
Suppose that $ a(\k)$ is a $\C^3$ valued function on $\R^3$ such that 
\begin{align}
&a.\ \ \ \ \k\cdot a(\k) = 0\  \text{for all} \  \k \in \R^3\ \ \text{and}     \label{pw31} \\
&b. \ \ \ \ \int_{\R^3} |a(\k)|^2 d^3\k / |\k| < \infty .     \label{pw32}
\end{align}
Define  
\begin{align}
A(x,t)& =\int_{\R^3} \(a(\k) e^{i(\k\cdot x - |\k|t)} +  \overline{a(\k)} e^{-i(\k\cdot x - |\k|t)}\) d^3\k/|\k|.\label{pw33}\\
B(x,t) &= curl\ A(x, t)  .  \label{pw34}\\
E(x,t) &= -(\p/\p t) A(x,t). \label{pw35}
\end{align}
Then $A,B$ and $E$ have 
 divergence zero  for each $t$ and 
 $B(\cdot), E(\cdot)$ is a 
 solution to Maxwell's equations  \eref{me}.
 Moreover for all real $t$ there holds
\begin{align}
\| \{ B(t), E(t)\}\|_{bw}^2 &=\| A(t)\|_{H_{1/2}}^2  + \| E(t)\|_{H_{-1/2}}^2  \notag\\
                            &= 2(2\pi)^3\int_{\R^3} |a(\k)|^2 d^3\k / |\k|.\ \ 
                            \label{pw36c}                             
\end{align}
Conversely, suppose that $B_0$ and $E_0$ are 
 divergence free vector fields on $\R^3$ 
and each is in $H_{-1/2}(\R^3)$. Then there is
a unique solution $B(t), E(t)$ to Maxwell's equations \eref{me} 
  and a unique
$($up to a set of measure zero$)$ function $a(\cdot):\R^3 \rightarrow \C^3$ satisfying 
\eref{pw31} and \eref{pw32} such that $B(t)$ and $E(t)$ are given by \eref{pw33} -  \eref{pw35}.
\end{theorem}
      \begin{proof}  
      Upon differentiating under the integral in \eref{pw33} we see from Lemma \ref{lempw} 
       that $A(t),B(t)$ and $E(t)$
      all have divergence zero and the pair $\{B(t), E(t)\}$ satisfies Maxwell's equations.

     For the proof of \eref{pw36c} observe first that, with $C = curl$ we  have \linebreak
     $ \|B(t)\|_{H_{-1/2}} = \|A(t)\|_{H_{1/2}}$
 because $\|B(t)\|_{H_{-1/2}} = \| C\, A(t)\|_{H_{-1/2}}=$ \linebreak 
 $  \| |C|^{-1/2} CA(t)\|_{L^2}    = \| |C|^{1/2} A(t)\|_{L^2}$.
 It suffices therefore to prove the 
 second equality in line \eref{pw36c}.    Moreover it suffices to prove it
 just at $t=0$ because
 replacing $a(\k)$ by $a(\k)e^{-i|\k|t}$ in the $t =0$ equality gives the equality at time $t$.
       Let
 \begin{align} 
 \alpha(\k) &= a(\k) + \overline{a(-\k)}\ \ \ \ \ \ \ \ \ \ \ \text{and}  \label{pw40}\\
 e(\k) &= i |\k|\Big( a(\k) - \overline{a(-\k)}\Big).                        \label{pw41}
 \end{align}
 Then from \eref{pw33}  and \eref{pw35} we find
 \begin{align}
 A(x,0) & = \int_{\R^3} \frac{\alpha(\k)}{|\k|} e^{i \k\cdot x} d\k . \label{pw44}\\
 E(x,0) & =  \int_{\R^3} \frac{e(\k)}{|\k|} e^{i \k\cdot x} d\k .\label{pw45}
 \end{align}
 Therefore
 \begin{align}
 \| A(\cdot, 0)&\|_{H_{1/2}}^2 + \| E(\cdot, 0)\|_{H_{-1/2}}^2                         \notag\\
 &= (2\pi)^3 \int_{\R^3} \Big(\Big| |\k|^{1/2} \frac{\alpha(\k)}{|\k|}\Big|^2 
                          +\Big| |\k|^{-1/2} \frac{e(\k)}{|\k|} \Big|^2\Big) d^3\k             \notag\\
 & =(2\pi)^3\int_{\R^3} \Big( | a(\k) + \overline{a(-\k)}|^2 + | a(\k) - \overline{a(-\k)}|^2\Big) d^3\k/|\k| \notag \\
 &=  (2\pi)^3\int_{\R^3} \Big(|a(\k)|^2 + |\overline{a(-\k)}|^2 \Big) d^3\k/|\k| \notag   \\
 & =(2\pi)^3 2 \int_{\R^3} |a(\k)|^2  d^3\k/|\k|.                            \label{pw47}
 \end{align}
 This proves the equality in line \eref{pw36c}.
 
    Conversely, suppose that $B_0$ and $E_0$ are given, that $div\ B_0 = div\ E_0 =0$ and that
$\|B_0\|_{H_{-1/2}}^2 + \| E_0\|_{H_{-1/2}}^2 < \infty$. Let $A_0 = ( curl)^{-1} B_0$. This defines $A_0$
as a potential with $div\, A_0 =0$ and
        $\| A_0\|_{H_{1/2}} = \|B_0 \| _{H_{-1/2}} < \infty$.
 Define $\alpha(\k)$ and $e(\k)$ from $A_0$ and $E_0$ by \eref{pw44} 
and \eref{pw45} respectively.   Since $A_0$ and $E_0$ are ``real'', i.e. both take their values
in $\R^3$, we have  $\alpha(-\k) = \overline{\alpha(\k)}$ 
and $e(-\k) = \overline{e(\k)}$. With \eref{pw40} and \eref{pw41} in mind, define
\begin{align}
a(\k) = (1/2)\Big(\alpha(\k) +\frac{e(\k)}{i|\k|}\Big).
\end{align}
From the hermiticity of $\alpha(\cdot)$ and $e(\cdot)$ it follows that
\beq
\overline{a(-\k)} = (1/2)\Big(\alpha(\k) - \frac{e(\k)}{i|\k|}\Big)
\eeq
and from this the relations \eref{pw40} and \eref{pw41} follow. The computation \eref{pw47} now shows
that
\beq
\|A_0\|_{H_{1/2}}^2 + \| E_0\|_{H_{-1/2}}^2 = (2\pi)^3 2 \int_{\R^3} |a(\k)|^2 d^3\k/|\k|
\eeq
Hence \eref{pw32} holds for the constructed function $a(\k)$. Of course \eref{pw31} follows 
because $A_0$ and $E_0$ are divergence free. The solution $B(t), E(t)$ with initial values $B_0, E_0$
can now be constructed from  $a(\k)$ by \eref{pw33} -\eref{pw35}.
\end{proof}

\begin{remark}\label{rmk64}{\rm (Integral formula for $\|\{B, E\}\|_{bw}^2$) 
Since the Laplacian ($\Delta \equiv \n^2$) on vector fields $u(x)$ over $\R^3$ is given by 
$-\Delta u(x) = curl^2 u(x) -grad\, div\, u(x)$ we may write
$-\Delta u(x) = curl^2 u(x)$ in the space of divergence free vector fields. 
 The space of divergence free vector fields 
 is invariant under $\Delta$ and  consequently $|curl| u= (-\Delta)^{1/2}u$ on this space.
 One   therefore has on this space $|curl|^{-1} =  (-\Delta)^{-1/2}$, which, by Fourier transform,
  is easily seen to be   given by convolution by $const. (1/|x|^2)$.
 Since $\|u\|_{H_{-1/2}}^2 = (|C|^{-1}u,u)_{L^2(\R^3)}$ on the space of divergence free 
 vector fields  $u$,   the Bargmann-Wigner norm \eref{n6} can be written
 \begin{align}
\|\{B, E\}\|_{bw}^2 = const. \int_{\R^3}\int_{\R^3}\frac{B(x)\cdot B(y) + E(x)\cdot E(y)}{|x-y|^2} dx dy. \label{n64}
\end{align}
 This representation of the Lorentz invariant norm for the electromagnetic field was derived in
 \cite{G5} and used there to show that the norm is invariant under the 15 dimensional conformal group of
 Minkowski space, which contains transformations to uniformly accelerated coordinate systems. 
 (See e.g.  \cite[footnote 2]{G5}).
 It was shown in \cite{G5} 
  that each of the two terms in \eref{n64} is separately invariant under the inversion
 $x\mapsto x/|x|^2$ of $\R^3$, which, together with dilations of $\R^4$ and the inhomogeneous
  Lorentz group, generate the conformal group. 
  Unitarity of the representation of the conformal group, as opposed to just norm invariance, 
   requires also proof of invariance
  of the complex structure, which  will be    
  described on the configuration space side in the proof
  of Theorem \ref{thmhel1}.  
   Invariance of the complex structure
   under the   conformal    group was first proved   in \cite{JV77}. 
 }
\end{remark}

\subsection{Helicity  and sgn(curl) for solutions of finite action}     \label{secsgncurl}

Our goal in this subsection is to take a step away from the plane wave decomposition \eref{pw33}
by formulating helicity entirely on the position space side with the help of the operator curl. 
This is reflected explicitly in the equivalence between assertions 1.) and 2.) in the next theorem

\begin{theorem}\label{thmhel1} 
Suppose that  $B(t), E(t)$ is  a solution to Maxwell's equations \eref{me} with finite
Bargmann-Wigner norm. Let \eref{pw33} - \eref{pw35} be its plane wave decomposition. Then
the following are equivalent.

1.$)$ Every plane wave in its 
plane wave decomposition 
has positive helicity.

2.$)$  $A(0)$ and $E(0)$ are in the positive spectral subspace of $curl$.

3.$)$  $B(t)$ and $E(t)$ are both in the positive spectral subspace of $curl$ for some $t$.

4.$)$ $B(t)$ and $E(t)$ are both in the positive spectral subspace of $curl$   for all $t$

5.$)$   $A(t)$ is in the positive spectral subspace of $curl$ for all $t$.

\noindent
These statements also hold with positive replaced by negative everywhere.
\end{theorem}

\begin{proof} 
The space of functions $a(\k)$  on momentum space satisfying \eref{pw31} and \eref{pw32}
represents Bargmann and Wigner's original description of the one-particle photon space, \cite{BW}. 
Denote it by $BW_{mom}$. 
On the position space side 
denote by BW the space of divergence free real fields $B,E$ over $\R^3$ with finite norm \eref{n6}.
Define 
\beq
\Ws: BW_{mom}\rightarrow BW    \label{sgn10}
\eeq
by 
$\Ws: a(\k)\mapsto \{B(x, 0), E(x,0)\}$ where $B$ and $E$ are given by \eref{pw33} - \eref{pw35}.
Then Theorem \ref{thmpwd} shows that $\Ws$ is a norm preserving  transformation 
from $BW_{mom}$ 
 onto the direct sum space $BW$ if one uses $2(2\pi)^3\int_{\R^3} |a(\k)|^2 d^3\k/|\k|$ for the 
norm squared on $BW_{mom}$. 
Actually, since there is no hermiticity condition on $a(\k)$, $BW_{mom}$ 
 is a complex Hilbert space
under the usual operation of multiplication by  $i\equiv \sqrt{-1}$.
$BW$ itself has a natural complex structure $j$ given by
$j \{B,E\} = \{sgn(C)E, - sgn(C) B\}$, where $sgn(C) = |C|^{-1} C$. The latter complex structure
is more often expressed in terms of  the pair $\{A, E\}$ in the form 
$j\{A, E\} = \{|C|^{-1} E, -|C| A\}$,
which is an equivalent description since $B = C A$. $BW$ is therefore also a complex Hilbert space.
 Comparing \eref{pw40} and \eref{pw41} with \eref{pw44}
and \eref{pw45} one sees easily that 
$\Ws ia = j\Ws a$.
$\Ws$ is therefore unitary.
Furthermore  \eref{hel36},  \eref{pw33}-\eref{pw35}, 
show that $\Ws$ intertwines $curl$ with the operator $C_{(\cdot)}$, which is multiplication by $C_\k$.
\begin{align}
curl\, \Ws = \Ws C_{(\cdot)} . 
\label{sgn11}
\end{align}
Here we are writing $curl\, \{B,E\} = \{ curl\, B, curl\, E\}$.

Now suppose that $B(t), E(t) $ is  a solution to Maxwell's equations \eref{me} with finite
Bargmann-Wigner norm and with plane wave decomposition given by \eref{pw33} - \eref{pw35}.
Suppose also that it  has only positive helicity plane waves in its plane wave decomposition.
According to Definition \ref{defHP}  $C_\k a(\k) =|\k| a(\k)$ for 
 every $\k \ne 0$ in $\R^3$. Therefore
$a(\cdot)$ is in the positive spectral subspace for the operator multiplication by $C_\k$. Since
$\Ws$ is unitary, $\{B(0), E(0)\}$ is in the positive spectral subspace of $curl$ on BW, (which is easily seen
to be equivalent to saying that $B(0)$ and $E(0)$ are each in the positive spectral subspace of $curl$.)
Since $|C|: H_{1/2}\rightarrow H_{-1/2}$ is an orthogonal transformation and commutes with curl it preserves
all spectral properties of curl. In particular $A(0)$ is also in the positive spectral subspace of curl because
$B_0 = curl\, A_0$.
Therefore statement 1.) implies statement 2.). 
        Conversely, if $A(0)$ and $E(0)$ are in the positive spectral subspace of $curl$ then
 $B(0)$ is in the positive spectral subspace of curl and so therefore is the pair $\{B(0), E(0)\}$. 
        By \eref{sgn11} and the unitarity of $W$ it follows that $a(\cdot)$ is in the positive spectral subspace
        of the operator of multiplication by $C_\k$. Hence $C_\k a(\k) = |\k| a(\k)$ for all $\k \ne 0$.
Therefore statement 1.) holds.  
 This proves the equivalence 
 of statements 1.) and 2.). The proof shows that the pair $\{A(0), E(0)\}$ could be replaced by the pair
 $\{ B(0), E(0)\}$ in the statement of item 2.).  Now for any fixed $t$ the map
  $\Ws_t:a(\cdot)\mapsto \{B(t), E(t)\}$ is also unitary and intertwines $curl$ and $C_\k$. Therefore the proof
  of equivalence of 1.) with  2.) also applies to the proof of equivalence  of 1.) with 3.) and 1.) with 4.).
  
  Concerning the condition 5.), we see  that if 1.) holds
   then by 2.) $A(0)$, and similarly $A(t)$, is in the positive spectral subspace of curl in its
own space $H_{1/2}$. Conversely, if 5.) holds then  $(d/dt)A(t)$ is in the positive spectral subspace
of curl. In particular $A(0)$ and its time derivative $ - E(0)$ are in the positive
 spectral subspace of curl  and we can apply  2.) to see that 1.) holds.
\end{proof}

\begin{remark}{\rm The operator of multiplication by $C_\k$ on $BW_{mom}$  
 decomposes $BW_{mom}$ 
 into
 two orthogonal subspaces given respectively by the two identities by  $C_\k a(\k) = |\k| a(\k)$  
 and    $C_\k a(\k) =- |\k| a(\k)$. 
 These are the two spectral subspaces of curl   on the momentum space side. 
}
\end{remark}

\subsection{Configuration space, phase space and helicity}     \label{seccpe}  
The Lorentz invariant norm \eref{pw36c} has a geometric interpretation
 based on  
 the canonical  formalism for the electromagnetic field:  
 As in  \cite[Chapter 14]{BD}, 
  we take configuration space  for the free electromagnetic field 
 to be 
 a set 
  of divergence free   potentials $A$ on $\R^3$ with size of $A$ yet to be specified.
  This corresponds to the radiation gauge. 
 The momentum canonically conjugate to $A$ is $-\dot A$, which is $E$. 
 See   \cite[ Eq. (14.10)]{BD} for a derivation of this from the Lagrangian formalism.
 The canonical formalism therefore  dictates that the pair $A, E$ is a point in phase
 space. 
 Thus if $A$ lies in some 
 configuration space $\cg$ (still to be determined, quantitatively) 
and $T_A(\cg)$ denotes the tangent space to $\cg$ at $A$,  then $ E$  is a point
 in its dual space $T_A^*(\cg)$, as is customary for momentum in classical mechanics. 
 In the next definition we will make a quantitative choice for $\cg$ and show that the norm
on phase space 
which is automatically induced by  this choice of configuration space 
 is   the Bargmann-Wigner norm \eref{pw36c}.
 
\begin{definition}\label{defconfig}{\rm  (Configuration space)  Let 
\begin{align}
\cg &= \{\text{real valued 1-forms $A$ on $\R^3$ such that}     \label{he199}\\
       & \qquad      a. \ \ div\, A=0                                             \label{he200}\\
       &  \qquad      b.\ \ \| A\|_\cg := 
                     \|\, |C|^{1/2} A\|_{L^2(\R^3)} < \infty \}.           \label{he201}
\end{align}
Here $C = curl$ as before. Define also
\begin{align}
\cg_+ &= \{ A \in \cg: A \ \text{is in the positive spectral subspace of}\ \ curl\}   \label{he202} \\
\cg_- &= \{ A \in \cg: A \ \text{is in the negative spectral subspace of}\ \ curl\} .  \label{he203}
\end{align}
Then $\cg$ is a real Hilbert space and
\begin{align}
\cg = \cg_+ \oplus \cg_-    .   \label{mp5} 
\end{align}
The subspaces intersect only at $A =0$ because curl does not have a zero eigenvalue in $\cg$.
}
\end{definition} 
Since $\cg$ is a vector space  
its tangent space at a point $A$ 
 can be identified with $\cg$ itself. The dual space  can therefore be identified with $\cg^*$:
 $T_A^*(\cg) \equiv \cg^*$. Consequently the entire phase space can be identified as
 \beq
  T^*(\cg) \equiv  \cg \oplus \cg^*  .          \label{he210}
  \eeq
 The electric field $E$ is to be identified with an element of
 the dual space   via the pairing 
  \beq
\< E, A\> = \int_{\R^3} E(x)\cdot A(x) d^3x,      \label{I11}
\eeq
which is the field theoretic analog of $\sum_{i=1}^n p_i dq_i$, when one identifies $T_A(\cg)$ with $\cg$.
Combined  with our choice of norm \eref{he201}, this pairing induces on $E$ the
 norm  $\|E\|_{\cg^*} = \|\, |C|^{-1/2} E\|_{L^2(\R^3)}$. 
  Since $\cg$ involves the restriction $div\, A =0 $, so does the dual space. We can therefore make the 
 natural identification 
\begin{align}
\cg^* &=\{\text{real valued vector fields 
$E$ on $\R^3$ such that} \label{he199e}\\
& \qquad      a. \ \ div\, E=0                                             \label{he200e}\\
       &  \qquad      b.\ \ \| E\|_{\cg^*} := 
                     \|\, |C|^{-1/2} E\|_{L^2(\R^3)} < \infty \}.           \label{he201e}
\end{align}
In the canonical formalism a  point in  phase space 
 is thereby specified by a pair $\{A, E\}$ with $\|A\|_\cg^2 + \|E\|_{\cg^*}^2 < \infty$ and  $div\, A = div\, E =0$.
 
We see from \eref{pw36c} at $t=0$ that this is precisely the Bargmann-Wigner norm 
on the initial data when we take $B = curl\, A$.
\begin{align}
\|\{B, E\}\|_{bw}^2 = \|A\|_\cg^2 + \| E\|_{\cg^*}^2, \ \ \ div\, A =0,\ \ div\, E =0. \label{he204}
\end{align}
Since $C$ commutes with $|C|$ the annihilator in $\cg^*$ of the positive spectral subspace of
curl in $\cg$ is the negative spectral subspace of curl in $\cg^*$. Similarly with positive 
and negative reversed.                  
Consequently we have the natural identifications    
\begin{align}
\cg_+^* &=\{E \in \cg^*: E\ \text{is in the positive spectral subspace of}\ curl\}.  \label{he205}\\
\cg_-^* &=\{E \in \cg^*: E\ \text{is in the negative spectral subspace of}\ curl\}. \label{he206}
\end{align}
Analogous to the identification \eref{he210} we therefore have the identifications
\begin{align}
T^*(\cg_+) = \cg_+ \oplus \cg_+^* , \  \ 
T^*(\cg_-) = \cg_- \oplus \cg_-^*, \ \ T^*(\cg) = T^*(\cg_+) + T^*(\cg_-) .       \label{he207}
\end{align}

        \begin{theorem} \label{thmhel2}  
        Suppose that the pair $\{A, E\}$ is a point in
   phase space $T^*(\cg)$. 
Let $A(t), E(t)$ be the unique
solution to Maxwell's equations with initial data $A(0) = A, -\dot A(0) = E(0) = E$.
Then the plane wave decomposition of  the solution consists of 

a$)$   positive helicity plane waves if and only if $\{A, E\}$ is in $T^*(\cg_+)$.

b$)$  negative  helicity plane waves if and only if $\{A, E\}$ is in $T^*(\cg_-)$.
\end{theorem}
       \begin{proof} The assumption that $\{A, E\}$ lies in $T^*(\cg)$ is equivalent, by \eref{he204},
        to the assumption that the pair
       has finite Bargmann-Wigner norm.
 In view of \eref{he202}, \eref{he203}, \eref{he205} and \eref{he206},
 the present  theorem just restates the equivalence
 between items 1.) and 2.) in Theorem \ref{thmhel1}.
\end{proof}

\begin{remark} \label{rmkpm}{\rm  
 (Primacy of the configuration space decomposition)
The helicity character of a solution to Maxwell's equations is not
determined by the value of $A(0)$ alone but also requires knowledge of $\dot A(0)$.
 To conclude,    for example, 
that a solution has only positive helicity plane waves in 
its plane wave decomposition it is not  sufficient to know that $A(0)$ is in the
 positive spectral subspace of curl. One must know also that $\dot A(0)$ is in the
  positive spectral subspace of curl. 
And yet the decomposition \eref{mp5} of $\cg$ is
responsible for determining helicity
in the sense that once this decomposition is known it automatically determines
 the decomposition of $T^*(\cg)$ as in \eref{he207},
allowing Theorem \ref{thmhel2} to be applied. The decomposition of $T^*(\cg)$ thereby
 plays a secondary 
role compared  to the decomposition \eref{mp5}. This is important to observe for two disparate reasons.
 First, in the quantized theory, where  a state of the field is given by a function $\psi$ on $\cg$, 
 the Heisenberg field expectations $(\Bs(x,t)\psi, \psi)$ and $ (\Es(x,t)\psi, \psi)$
 in the state $\psi$  
  will be shown   in \cite{G76}
 to  have only positive helicity plane waves
  in their plane wave decomposition when 
  $\psi$ depends only on $\cg_+$. 
 Second, in the non-abelian theory, where  the configuration space is 
     a non-linear
     manifold, the electromagnetic decomposition \eref{mp5}
     has  a natural analog with properties
     that well imitate those of the electromagnetic case. This will be shown in Section \ref{secdecomp}.
    }
\end{remark}

\begin{remark}\label{remhist}{\rm (A bit of history)  Immediately  after 
  Heisenberg and Pauli published their fundamental  paper \cite{HP1} in 1929,  
 describing the quantized electromagnetic field,
  Landau and Peierls \cite{LP1}, in 1930,  rederived  some of the formalism of Heisenberg
and Pauli, emphasizing 
 the role of the electromagnetic configuration space.  
In the formulation of Landau and Peierls $E$ and $B$ were not independent. 
They introduced a norm  equal to one of the two terms in \eref{n6} (or equivalently \eref{n64}) 
 for the normalization
  condition for their configuration  space. See e.g. \cite[Equations (7)-(10)]{LP1}.  
  It is now  understood 
  that in the canonical formalism $A$ and $E$ are to be regarded 
   as independent variables in phase space. 
  So their suggestion to allow either term in \eref{n64} as the norm on their configuration space obscured
  the distinction between configuration space and phase space. 
 Interestingly, they arrived at their norm, not by citing Lorentz invariance, but by
  arguing that a unit vector in this particular norm represents a single photon. 
 An updated version of their ``number of photons''
  argument for this norm,  leading to the full phase space norm \eref{n64}, can be
   found in Jackson  \cite[Problem 7.30]{Jac}. 
}
\end{remark}

\section{Equivalence of helicity and Euclidean self-duality in electromagnetism} 
\label{sechsd}

\subsection{The Poisson action and Maxwell-Poisson semigroup}
\label{secpmp}
The helicity decomposition \eref{mp5} and it's associated decomposition
of phase space \eref{he207} are 
based on use of the spectral decomposition of the operator curl and were shown to be
 equivalently characterizable 
in terms of plain wave expansions. (cf. Theorems \ref{thmhel1} and \ref{thmhel2}.)
But neither plain wave expansions nor the operator curl mesh well with non-abelian gauge fields.
In this section we are going  to show that 
 the helicity  decomposition of configuration space  
 can be  reproduced 
 without using either plane wave expansions or the spectral decomposition of curl. 
 Instead, we will extend the gauge potential $A$ to a half-space in $\R^4$ and 
  show that the  decomposition  is equivalent to (anti-)self-duality of the extension.  
This  method for implementing the decomposition \eref{mp5} in electromagnetism 
will be shown in succeeding sections
to go over to non-abelian gauge fields. 

 To simplify otherwise cumbersome calculations we will make use henceforth of differential form notation.
 Recall that if $(u_1, u_2, u_3)$ is a vector  field on  $\R^3$ then the corresponding differential form
 is given by $u = \sum_{i=1}^3u_i dx^i$, the Hodge $*$ operator is given on 1-forms by 
 $*u = (1/2)\sum_{i,j,k}\epsilon_{ijk}u_i\, dx^j\wedge dx^k$ and  on 2-forms by 
 $*\sum_{jk}v_{jk} dx^j\wedge dx^k = \sum_{ijk} \epsilon_{ijk} v_{jk} dx^i$. The exterior derivative is given by
 $ du =\sum_{ji} (\p u_i/\p x_j) dx^j\wedge dx^i$ and the action of $curl$ on 1-forms is given by $curl\, u = *du$.
  
\begin{notation}{\rm  Denote by  $\R^4_+$ the Euclidean half space $\R^3 \times [0,\infty)$ with coordinates
$(x,s)$.   Let $\as$ be  a real valued 1-form on $\R^4_+$ in temporal gauge. That is, 
\begin{align}
\as(x,s) = \sum_{j=1}^3\as_j(x,s) dx^j, \ \ \ \ x \in \R^3, \ \ s \in [0, \infty).          \label{mp10}
\end{align}
 Then the four dimensional curvature of $\a$ is given by
\begin{align}
 F    = ds\wedge \as'  +d\as,                                       \label{mp12}
\end{align}
where $'=(\p/\p s)$ and  $d$ is again  the three dimensional exterior derivative operator.

We wish to minimize the  functional 
\beq
\int_{\R^4_+}  | F(x,s)|^2 dx ds                 \label{mp15}
\eeq 
subject to the initial condition
\begin{align}
\as(x,0) = A(x),\ \ \ x \in \R^3,                                                          \label{mp16}
\end{align}
wherein $A$ is a given real valued 1-form on $\R^3$. Observe first that
\begin{align}
|F(x,s)|_{\L^2}^2 =  |\as'(x, s)|_{\L^1}^2 + |d\as(x,s)|_{\L^2}^2    \label{mp17}
\end{align}
for each $(x, s)$ because the first 2-form on the right side of \eref{mp12} is orthogonal  to the second one.
In \eref{mp17} $\L^1$ denotes, as usual, the three dimensional space of linear functionals on $\R^3$
and $\L^2$ denotes the three dimensional space of skew symmetric bilinear functionals on $\R^3$.
The variational equation for the minimization problem can be derived in the  usual way  as follows.
Let $u(x,s) = \sum_{j=1}^3 u_j(x,s) dx^j$ with each coefficient lying in $C_c^{\infty} (\R^3\times (0, \infty))$.
Then we find
\begin{align}
\p_u\int_{\R^4_+}&\Big\{ |\as'(x,s) |^2  + |d\as(x,s)|^2\Big\} dxds \notag\\
&= 2\int_{\R^4_+}\Big\{ -(\as''(x,s), u(x,s)) + (d^*d \as(x,s), u(x,s))\Big\} dx ds
\end{align}
after doing an integration by parts in both time and space.
The variational equation is therefore
\begin{align}
\as'' = d^*d \as.       \label{mp20}
\end{align}
This is not quite Poisson's equation because half of the Laplacian, $-\Delta = d^*d + dd^*$,
is missing from the right hand side. But 
it is intimately related to Maxwell's theory
and we will refer to \eref{mp20} as the Maxwell-Poisson equation. It is easily  solved. 
Since $d^* = *d*$ on 2-forms over $\R^3$ we see that  $d^*d = (*d)(*d) = curl^2$. The equation \eref{mp20}
may therefore be written as $\as''(s) = C^2 \as(s)$ with $C = curl$. The function $\as(s) = e^{sC}A$ clearly solves this equation 
as does also $\as(s) = e^{-sC}A$. But if $A$ is in the strictly positive spectral subspace of curl then
$\|\as'(s)\|_{L^2(\R^3)}$ grows exponentially in the first case, making \eref{mp15}  infinite, while 
if $A$ is in the strictly negative spectral subspace of curl then \eref{mp15} is infinite in the second case.
The solution in both cases is given, for arbitrary $A \in L^2(\R^3, \L^1)$, by  
\begin{align}
\as(s) := e^{-s|C|} A,\ \ \ \ s \ge 0.       \label{mp30}
\end{align}
It solves the Maxwell-Poisson equation \eref{mp20} because  $\as'' = |C|^2 \as = C^2 \as = d^*d \as$
and also satisfies the initial condition $\as(0) = A$.
Under reasonable growth restrictions on $\as(\cdot)$ the solution given by \eref{mp30} is unique.
 We will  
always assume  uniqueness. The size of the minimum in \eref{mp15} will be discussed in the following.
It bears emphasizing that we are not making an assumption that $A$ is in Coulomb gauge in this discussion.
}
 \end{notation}

\begin{terminology}{\rm  
The unique solution to the Maxwell-Poisson equation \eref{mp20}
with initial value $A$ will be referred to as 
the {\it Poisson extension of $A$}. We define the {\it Poisson action}
of $A$ by
\begin{align}
\P(A) =   \int_0^\infty \( \|\as'(s)\|_2^2 + \| d\as(s)\|_2^2 \) ds,  
\label{mp33}
\end{align}
where $\a$ is the Poisson extension of $A$. 
We are only interested in those potentials $A$ for which $\P(A) < \infty$.
We  reiterate that $\P(A)$ is the minimum of the  integral in \eref{mp33} 
 subject to the condition that $\as(0) =A$.
}
\end{terminology}

\begin{remark}{\rm The Poisson action of a 
gauge potential $A$ can be computed explicitly.
We need not assume that $A$ is in Coulomb gauge. From \eref{mp30} we find
\begin{align}
\as'(s) &= -|C| e^{-s|C|} A     \notag \\
\|\as'(s)\|_2^2 &= \|\ |C| e^{-s|C|} A\|_2^2 = (|C|^2 e^{-2s|C|} A, A)     \notag\\
\|d\as(s)\|_2^2 & =\|*d\as(s)\|_2^2 =\| C e^{-s|C|} A\|_2^2 = (|C|^2 e^{-2s|C|} A, A). \notag 
\end{align}
Therefore
\begin{align}\P(A)  &= 2\int_0^\infty (|C|^2 e^{-2s|C|} A, A) ds \notag \\
& = (|C| A, A) .                                                                  \label{mp50}
\end{align}
A reader concerned about the lack of exponential decrease in \eref{mp30}  
 when $A$ is in the null space
of curl should observe that in this case the integrand in \eref{mp33} is identically zero.
}
\end{remark}

\begin{remark}{\rm The identity \eref{mp50} holds even if $A$ is longitudinal,
 that is, $A = d\lambda$ for some real scalar function $\lambda$. In this case $A$ is in the null space 
 of $C$ and therefore also in the nullspace of $|C|$. It follows then from \eref{mp50} that
 \begin{align}
 \P(d\lambda)= 0\ \ \text{for any function}\  \lambda:\R^3 \rightarrow \R.    
 \label{mp34}
\end{align}
 This could also be derived  directly from the definition \eref{mp33} because the solution to the
 Maxwell-Poisson equation \eref{mp20} is easily verified to be $\as(s) = d\lambda$.
 
 On the other hand, for $A$ in Coulomb gauge,  \eref{mp50} and Definition \ref{defconfig} show that  
\begin{align}
\P(A) = (|C|A, A)_2 = \|A\|_\cg^2 \ \ \ \ \text{if}\   div\, A = 0.               \label{mp35}
\end{align}
In particular if $A$ is in Coulomb gauge then $\P(A) < \infty$ if and only if $A \in H_{1/2}(\R^3)$.
  }
 \end{remark}
 
 \begin{remark}{\rm A general gauge potential $A$ can be decomposed into its longitudinal part
  and transverse part: 
  \begin{align}
  A = A_{long} + A _ {trans},\ \ \ \ A_{long} = d\lambda,\ \ \ \ div\, A_{trans} =0.
  \end{align}
  Its Poisson extension is then   given by
  \begin{align}
  \as(s) = A_{long} + e^{-s|C|} A_{trans} . \label{mp60}
  \end{align}
  Since $|C|$ has zero nullspace among the transverse potentials, the second term goes to zero
  as $s\to \infty$. Hence
  \begin{align}
\lim_{s\to \infty} \as(s) = A_{long} .               \label{mp61}
\end{align}
 The Maxwell-Poisson equation therefore  filters out the transverse part of $A$
and leaves the longitudinal part (i.e. pure gauge part) in the limit. For a non-abelian gauge
 field the  
 pure gauge part cannot be separated out from the initial potential $A$. 
Instead we will use the non-abelian analog of the Maxwell-Poisson equation to 
separate out a ``pure gauge piece'' of the initial data
by an analog of \eref{mp61}.
  }
  \end{remark}

\begin{remark} {\rm 
If $A$ is in Coulomb gauge then \eref{mp35} shows that the first
 term in the Bargmann-Wigner norm \eref{he204}   
 is equal to the Poisson action of $A$. Since the full Bargmann-Wigner norm is then determined
 by the dual space norm as in \eref{he204}, the Poisson action actually determines the full Lorentz invariant
  norm on the initial data space for Maxwell's equations.
  }
\end{remark}

\subsection{Self duality and helicity for the classical electromagnetic field}    \label{secsdm}

The four dimensional Euclidean Hodge star operation $*_e$ 
is given on the four dimensional curvature $F$ defined in \eref{mp12} by 
\begin{align}
*_e\,  F =  ds\wedge *d\as + *\as' ,                    \label{mp22}
\end{align}
wherein $*$ denotes the three dimensional Hodge star operation.
Therefore the four dimensional curvature  $F$  
is self-dual if and only if
\begin{align}
ds\wedge \as' + d\as = ds\wedge *d\as + *\as'\ \ \ \text{(self-duality)}.       \label{mp23}
\end{align}
That is
\begin{align}
\as' = *d\as      \ \ \ \ \ \text{(self-duality)}.                                        \label{mp24}
\end{align}
(Keep in mind that $*^2 = Identity$ over $\R^3$.)
Similarly $F$ is anti-self-dual if and only if
\begin{align}
\as' = - *d\as     \ \ \ \ \text{(anti-self-duality)}.                                   \label{mp25}
\end{align}

It is clear that \eref{mp24} and \eref{mp25} each imply the Maxwell-Poisson equation \eref{mp20}

        \begin{theorem}\label{thmhsa}  
                 Suppose that $div\ A=0$ and that $A \in H_{1/2}(\R^3)$.
         Let $\as(s)$ be the Poisson extension of $A$. Then

$A$ is in the negative spectral subspace of $curl$ if and only if $\as(\cdot)$ is self-dual.
 
$A$ is in the positive  spectral subspace of $curl$ if and only if $\as(\cdot)$ is anti- self-dual.

\end{theorem}

\begin{proof} Writing $C = curl = *d$ we may write \eref{mp24} and \eref{mp25} in the form
 \begin{align}
 \as' &= C \as \ \ \ \ \ \ \text{(self-dual)}         \label{mp40}\\
\as' &= - C\as \ \ \ \ \text{(anti-self-dual)}  .    \label{mp41}
\end{align}
If $A$ lies in the negative spectral subspace of curl then $e^{-s|C|}A = e^{sC}A$ and therefore the solution
\eref{mp30} reduces to $\as(s) = e^{sC} A$. Consequently $\as'(s) = C\as(s)$ and therefore
 $\as(\cdot)$ is self-dual. 
 
        Conversely, suppose that $\as(\cdot)$ is self-dual. Then by  \eref{mp40} we have
 $\as(s) = e^{sC}A$ and  therefore $\|\as'(s)\|_{L^2( \R^3)} = \|Ce^{sC}A\|_{L^2( \R^3)}$.
 If $A$ has a non-zero spectral component for 
$C$ in the spectral subspace $[\epsilon, \infty)$ with some $\epsilon >0$ then, by the spectral theorem,
$\|Ce^{sC}A\|_{L^2( \R^3)} \ge  c \epsilon \exp (\epsilon s)$ for some strictly positive constant $c$. 
 Therefore $\as(\cdot)$
does not have finite Poisson action and $A$ is not in $H_{1/2}$. So $A$ must lie in the spectral subspace $(-\infty, 0]$ for $C$
if $\as$ is self-dual. 

 A similar argument applies to show that the Poisson extension of $A$ is anti-self-dual if and
  only if $A$ lies in the positive spectral  subspace of curl.
  \end{proof}

\bigskip
\noindent
The following corollary is the main result of this section.
        \begin{corollary}\label{corasd}   
 Suppose that $A$ is in the electromagnetic configuration space $\cg$ and $\as$ is
its Poisson extension. Then 

$A$ is in $\cg_+$ if and only if $\as$ is anti-self-dual.

$A$ is in $\cg_-$ if and only if $\as$ is self-dual.
\end{corollary}
       \begin{proof} This just restates Theorem \ref{thmhsa} in view of  Definition \ref{defconfig}.
\end{proof}

\section{Yang-Mills fields} \label{secym}

Both the Lorentz invariant norm and helicity  for the electromagnetic field can  be characterized 
 in terms of either plane waves, the operator curl,  or the Maxwell-Poisson semigroup, as we saw in Sections \ref{secpw} and \ref{sechsd}.
    Since the classical Yang-Mills equations are non-linear,  there is no useful 
gauge invariant representation
of these fields in terms of plane waves. 
 In this section we are going to 
 describe  the Yang-Mills analog of the Maxwell-Poisson semigroup and then use it to
  produce both a Riemannian metric  
  on configuration space   and  a decomposition of 
 configuration space into positive and negative ``helicity'' submanifolds.  
 Here is an outline of the transition steps from electromagnetism to Yang-Mills.

{\it i})  The  norm on the (linear) configuration space of the electromagnetic field
 will be replaced by a Riemannian
 metric on the (non-linear) configuration space $\cg$ of the Yang-Mills field.  The Yang-Mills-Poisson
 equation will be instrumental in this construction. (Section \ref{seccsym}.)
 
 {\it ii}) The  helicity subspaces $\cg_\pm$ in the linear decomposition \eref{mp5} will be replaced  by
 submanifolds $\cg_\pm$  determined by (anti-)self duality of solutions to the Yang-Mills-Poisson equation.
 (Section \ref{seccpm}.)
 
 {\it iii}) The orthogonal projections onto the
 helicity subspaces $\cg_\pm$ in the linear decomposition \eref{mp5} will be replaced
  by flows along orthogonal vector fields  over the manifold $\cg$, yielding non-linear
  analogs of these projections. (Section \ref{secflows}.)

 {\it iv})    The spectral behavior of the operator curl in the electromagnetic helicity spaces
      will be replaced by a 
      corresponding behavior of the gauge covariant curl operator acting not on the
       configuration space $\cg$ but on the tangent
      spaces $T(\cg_\pm)$.      (Section \ref{secflows}.)

\section{The configuration space for Yang-Mills fields} \label{seccsym}
\subsection{The Poisson action and configuration space}  \label{secpacs}

Denote by $K$ a compact Lie group contained in the orthogonal or unitary group
 of a finite dimensional inner product space $V$. 
 $\kf$ will denote its Lie algebra, on which we assume given 
   an $Ad\, K$ invariant real inner product  $\<\cdot, \cdot\>_\kf$.
Let 
\begin{align}
\as(s) := \sum_{j=1}^3 \as_j(s) dx^j,\ \ \ s \ge 0     \label{py2}
\end{align}
be a 
$\kf$ valued 1-form on $\R^3$ for each (Euclidean) time $s \ge 0$. 
Here each coefficient $\as_j(x,s)$ lies in $\kf$.
Denote by
\begin{align} 
\bs(s) :=d\as(s) +   \as(s) \wedge \as(s)            
       \label{py3}
\end{align}
its three dimensional curvature and denote by  $\as'(s)$  its derivative with respect to $s$.
We may regard $\as$ as a 1-form on a half space of $\R^4$ in temporal gauge. 
That is, it has  no $ds$ component.  
Its four dimensional curvature is then given by
 \begin{align}
 F = ds\wedge \as'(s) + \bs(s).                          \label{py4}
 \end{align}
Since the two summands are mutually orthogonal at each point $(x,s)$ we have
\begin{align}
|F(x,s)|_{\L^2\otimes\kf}^2 = 
|\as'(x, s)|_{\L^1\otimes \kf}^2 + |\bs(x,s)|_{\L^2\otimes \kf}^2    \label{py5}    
\end{align}
in analogy with \eref{mp17}. 
We are interested in the minimization of the  functional 
\begin{align}
\int_0^\infty \Big( \| \as'(s)\|_{L^2(\R^3)}^2 + \| \bs(s)\|_{L^2(\R^3)}^2 \Big) ds          \label{py6}
\end{align}
subject to an initial condition $\as(0) = A$,
where $A$ is a given $\kf$ valued 1-form on $\R^3$. In view of \eref{py5} this functional of $\as$ is also
given by the familiar expression \eref{mp15} but with $F$ now given by \eref{py4}.

The Euler equation for this minimization problem is the Yang-Mills-Poisson equation
\begin{align} 
\as''(s) = d_{\as(s)}^* \bs(s), \ \ \ \ \ \ \as(0) = A.                     \label{pe6}
\end{align}
Here we are using $d_A$ to denote the gauge covariant exterior derivative, which is given on
any $\kf$ valued p-form on $\R^3$ by $d_A \w = d \w + [A\wedge \w]$, while 
$d_{\as(s)}^*: L^2(\R^3; \L^2\otimes \kf) \rightarrow L^2(\R^3; \L^1\otimes \kf) $  denotes the adjoint of $d_{\as(s)}$    
 in \eref{pe6}. The suggestive notation $[u\wedge v]$ for $\kf$ valued forms means to use commutator 
 of the coefficients with wedge product of the differentials. 
 For example $[(\sum_i u_idx^i)\wedge(\sum_j v_j dx^j)]  = \sum_{ij} [u_i, v_j] dx^i\wedge dx^j$.
 Later we will also need to use the adjoint of this product, given by $\<[u\lrc v], \w\> = \<v, [u\wedge \w]\>$,
 which must hold for each point in $\R^3$ and for all $\kf$ valued forms $\w$ such
  that $ deg\, u + deg\,\w = deg\, v$.  
  $\<\cdot, \cdot\>$ is the inner product on forms induced by the given inner product on $\kf$.

The Yang-Mills-Poisson equation \eref{pe6} is equivalent to the elliptic Yang-Mills equation
 for forms in temporal  gauge since 
$d_\as^{(*_e)}F =-\as'' + d_\as^*\bs$, where $*_e$ is the Euclidean Hodge star operator. 
This is an elliptic non-linear boundary value problem which has been investigated in ground breaking
work by A. Marini and T. Isobe. \cite{Ma2,Ma3,Ma4,Ma5,Ma6,Ma7,Ma11,Ma12,Ma13}. 
For a bounded open set in $\R^4$ they have proven  existence and non-uniqueness of solutions to the 
minimization problem if one allows  the solution to take on the prescribed boundary values in a weak
sense, namely equality only up to local gauge transformations.
     Our half space problem can be conformally transformed into their setting. 
     But we need a stronger notion of attainment of the boundary value at $s=0$, namely actual equality. 
 However  we can allow a singularity at the point of their boundary corresponding to  $s =\infty$.     
  Thus at the present time it is not inconsistent with their work to assume existence and uniqueness
  for solutions to the Yang-Mills-Poisson equation \eref{pe6} under some appropriate conditions on $A$.  
 We will assume throughout the remainder of this paper that for the potentials $A$ of interest to
 us  existence and uniqueness holds for the boundary value problem \eref{pe6} with finite value
 of the action functional in \eref{py6}. In Example \ref{inst} we will see how instantons give rise to such solutions.

 \begin{notation}{\rm (Poisson action)
 Define
 \begin{align}
 \P(A) = \int_0^\infty \Big( \| \as'(s)\|_{L^2(\R^3)}^2 + \| \bs(s)\|_{L^2(\R^3)}^2 \Big) ds,  \label{py8}
 \end{align}
wherein $\as$ is the solution to the boundary value problem \eref{pe6}. We will refer to $\as$ as
 the {\it  Poisson extension} of $A$ and to $\P(A)$ as the {\it Poisson action}  of $A$.
 }
\end{notation}

 A gauge function $g:\R^3\rightarrow K$ acts on a potential $A$ via the action 
 $A\mapsto A^g:= g^{-1}A g + g^{-1} dg$.
 The Poisson action is a gauge invariant function of $A$. That is,
 \beq
\P(A^g) = \P(A).    \label{py12}
\eeq
To see this, suppose that  
$\as$ is a solution to the 
Yang-Mills-Poisson equation \eref{pe6}. Then the function
\beq
\as^g(s) := g^{-1} \as(s) g + g^{-1} dg     \label{py10}
\eeq
is again a solution to \eref{pe6}, as we can see from the identities 
$(d/ds)^n \as^g(s) = g^{-1}(d/ds)^n \as(s) g, n = 1,2$ and 
$d_{\as^g(s)}^* \bs^g(s)  = g^{-1} \(d_{\as(s)}^* b(s)\) g$. Of course $\as^g(0) = A^g$. 
And since $\bs^g(s) = g^{-1} \bs(s) g$ we see also that  the functional in \eref{py8} has the same value
for $\as^g$. This proves \eref{py12}.

As in the electromagnetic case $\P(A)$ is zero on pure gauges. (cf. \eref{mp34}.)
Indeed, if $A$ is  a pure gauge, say $A = g^{-1} dg$, then its Poisson extension is given by
$\as(s) = g^{-1} dg$ for all $s \ge 0$ since both $\as'(s) =0$ and $\bs(s) =0$ for this function, which shows
that \eref{pe6} holds and also that 
\beq
\P(g^{-1} dg) =0.    \label{py13}
\eeq

We wish to use the function $\P(A)$ to construct the relativistically correct configuration space
as we did in the electromagnetic case. Nominally, the configuration space is given by $\cg = \A/\G$
for some appropriate set  $\A$ of gauge potentials on $\R^3$ and corresponding gauge group $\G$.
But \eref{py13} shows that the restriction $\P(A) < \infty$  provides no control over $A$ in some directions
(the longitudinal directions). 
 Defining $\A$ simply by the condition $\P(A) < \infty$ 
 would lead therefore to an unsatisfactory candidate for $\A$, a circumstance
which we avoided in the electromagnetic case  by imposing the Coulomb gauge on $A$.
       In our non-abelian case we want to
avoid such a gauge choice  because of the ever lurking Gribov ambiguity. We will instead
put an indirect size restriction on the 
gauge potentials  allowed into our space $\A$ 
by making use of the filtering action of the Yang-Mills-Poisson  
equation. 

Let 
\begin{align}
\| \w\|_{H_{1/2}} = \|(-\Delta)^{1/4} \w\|_{L^2(\R^3)}    \notag
\end{align}
for any $\kf$ valued 1-form on $\R^3$. Here  we are using the Laplacian $-\Delta := d^*d +dd^*$ on 
$\kf$ valued 1-forms over $\R^3$.
We have chosen the Sobolev index $1/2$ because in the electromagnetic case the size restriction
$\P(A) < \infty$ gives exactly $H_{1/2}$ gauge potentials when the Coulomb gauge is imposed,
 as we saw in \eref{mp50}.  We are going to impose an $H_{1/2}$ restriction on the ``filtered longitudinal
 part'' of $A$ as follows.

Suppose that $\as(s)$ is a solution to \eref{pe6} and 
 that 
 \begin{align}
 \as_\infty :=\lim_{s\rightarrow \infty} \as(s)
 \end{align}
  exists (for example in the sense of $L^2(\R^3)_{loc}$).
 We saw in \eref{mp61} that $\as_\infty$ exists in the electromagnetic case and is precisely
  the longitudinal part (i.e. the pure gauge part)  of the initial potential $A$. 

\begin{notation} \label{notPAC} {\rm ($\Ra$, $\A$ and $\cg$)
Let 
\begin{align}
\Ra = \{ &\text{solutions $\as$ to the differential equation in \eref{pe6} such that}\notag\\
&\qquad \P(\as(0)) + \|\as_\infty\|_{H_{1/2}}^2 < \infty\}.       \label{py14}
 \end{align}
 The gauge potentials of interest to  us will be the set of initial values of such solutions:
 \begin{align}
 \A = \{ \as(0): \as \in \Ra\}.           \label{py15}
 \end{align}
We will say that a gauge potential $A$ on $\R^3$ is in the {\it soft Coulomb gauge} if it lies in $\A$,
and in the  {\it asymptotic Coulomb gauge} if in addition $\as_\infty =0$.
 The examples we construct from instantons below will be in the asymptotic Coulomb gauge.

Having now chosen the space $\A$ in a way that parrots the electromagnetic case, the gauge group
$\G$ appropriate for this choice must be chosen so as to preserve $\A$ as a set and also
 to be isometric for the as yet to be determined Riemannian metric on $\A$.  
      If $\as$ is the Poisson extension of $A$ then $\as^g$ is the Poisson extension of $A^g$ as we have seen
      in the proof of  \eref{py12}. Since $\lim_{s\to \infty} \as^g(s) = (\as_\infty)^g$ we need
\begin{align}
\as_{\infty}^g \in H_{1/2} \ \ \text{whenever} \ \   \as_\infty \in H_{1/2}     \notag
\end{align}
in order for $\A$ to be a gauge invariant set under the gauge function $g$. 
In view of \eref{py10} we therefore need $g^{-1}dg$ to be itself in $H_{1/2}$.
With this as motivation we define
\begin{align}
\G = \{&\text{gauge functions $g:\R^3 \rightarrow K$  such that  $g^{-1}dg$ is in $H_{1/2}$} \notag\\
&\qquad\qquad\text{and  $g(x) \rightarrow I$ as $x \rightarrow \infty$}\}.  \label{py17}
\end{align}
The second condition reflects  the fact that we are interested in charge zero.
The relation between the behavior at infinity of a gauge function on the one hand and
 total charge on the other results from the well known
relation between gauge invariance and conservation of charge. This will be elaborated on in a future work.

It happens that $\G$  is a complete topological group in its natural metric under
 pointwise multiplication.   It is also the critical gauge group for three spatial dimensions in the sense 
 that for any $\epsilon >0$, the set of gauge functions for which $g^{-1}dg$ lies in the (locally) 
 larger space  $H_{1/2 - \epsilon}$   is no longer a group:  It is not closed under multiplication.
 See \cite{G70} for further properties and proofs.
 Combining the set $\A$ of ``$H_{1/2}$''  gauge potentials with the matching group $\G$ of
  ``$H_{3/2}$'' gauge functions,  we define the configuration space for the non-abelian gauge field to be
\begin{align}
\cg = \A / \G.                                  \label{py18}
\end{align}
}
\end{notation}
We will   make $\A$          
into a Riemannian manifold with a $\G$ invariant metric in the next two sections.
To this end we will   gather  information on the derivatives of $\P$ 
 in the next section. 
This produces, of course, a Riemannian metric on $\cg$ also.

\begin{remark}\label{remsoft}{\rm 
 In addition to the fundamental mathematical
 question as to the existence
and uniqueness of solutions to the Yang-Mills-Poisson equation \eref{pe6}, these definitions 
raise the   question as to whether
 the limit $\as_\infty$ actually exists
for some appropriate 
 class of gauge potentials $A$.  
 For the Yang-Mills heat equation a similar question has been addressed in many works.
 See e.g. 
 \cite{HT1,HT2,HTY},  
 \cite{CG2}, \cite{Ja5},   
\cite{OT5} and references therein.  
 }
\end{remark}

\subsection{Restriction  of instantons}         \label{secinst}

\begin{example}\label{inst}{\rm (Restriction of instantons) 
Finite action solutions to the Yang-Mills-Poisson equation can be constructed from instantons
by gauge transforming an
instanton  to the Euclidean temporal gauge and then restricting it 
to a half-space.
Take, for example,  $K = SU(2)$. The simplest instanton is given by
\begin{align}
A_\mu(x) =\s_{\mu,\nu}x^\nu /(|x|^2 +\rho^2), \ \ \ x = (\xb, s) \in\R^4. \notag
\end{align}
where $\s_{\mu,\nu} =-\s_{\nu,\mu} \in \kf =su(2)$ and $\rho >0$ and constant. 
 Its curvature is given by
\begin{align}
F_{\mu,\nu} =-4\s_{\mu,\nu}\frac{\rho^2}{(|x|^2 +\rho^2)^2}.   \notag
\end{align}
See e.g. \cite{BC1} or  \cite{At} or  \cite[Eq. (21.26)]{Shif} or \cite{VV}.

We want to gauge transform $A$ into the Euclidean temporal gauge, i.e., so that its new fourth
 component is equal to zero.
For a gauge function $g:\R^4\rightarrow SU(2)$ we have $(A^g)_4 =g^{-1} A_4 g + g^{-1} \p_s g$.
So we want to choose $g$ so that $A_4 + (\p_s g)g^{-1} =0$. We may write
\begin{align}
A_4(\xb, s) = \frac{\bf u}{s^2+a^2},    \notag
\end{align}
where ${\bf u} =\sum_{j=1}^3 \s_{4,j} \xb^j \in \kf$ and $a^2 = |\xb|^2 + \rho^2$. 
The solution which is the identity element of $SU(2)$ at $s =0$ is given by Bitar and Chang
\cite[Equ. (4.4)]{BC1} as $g_{bc}(\xb, s) =\exp(-({\bf u}/a) tan^{-1}(s/a))$.
For $0\le s < \infty$, $ g_{bc}(\cdot, s)$ lies in our  critical gauge group $\G$, 
but $g_{bc}(\cdot, \infty)$ does not.
We need to use the solution which is the identity element at $s =\infty$.
It is given by 
\begin{align}
g(\xb,s) = \exp\Big(\frac{{\bf u}}{a}\Big\{(\pi/2)- tan^{-1}(s/a)\Big\}\Big), \ \ \  (\xb,s) \in \R^4.   \notag
\end{align}
Then $(A^g)_4 = 0$ and 
\begin{align}
\as_j(\xb, s) &\equiv (A^g)_j(\xb, s)         \notag\\
&=g(\xb, s)^{-1} A_j(\xb,s) g(\xb, s) +g(\xb,s)^{-1} \p_jg(\xb,s),\   j=1,2,3.    \notag
\end{align}
Since $g(\xb, \infty) = I$ and $A_j(\xb, \infty) = 0$ we see that $(A^g)_j(\xb, \infty) =0$.
So $A^g$ is in Euclidean temporal gauge and also in asymptotic Coulomb gauge.
Since it is the four dimensional gauge transform of an instanton it satisfies 
$d_{A^g}^{{*_e}} F(A^g) =0$, which is Poisson's equation \eref{pe6} for $\as$.
Its curvature
is square integrable over all of $\R^4$ and therefore over the half-space $\R^4_+$.
Hence $\as$ has finite action and is in asymptotic coulomb gauge. In particular its initial
 value, $\as(0)$, lies in  $\A$.

Since, by definition, an instanton has finite action, this construction applies to any instanton
 as long as the time dependent gauge transformation
    $g(\xb,s)$ needed to transform the instanton to temporal gauge can be chosen to be the identity
    of the group $K$ at $s = \infty$.

}
\end{example}

\subsection{Derivatives of  the Poisson action. } \label{secderivs}
\begin{notation}\label{nottangent}{\rm (Variational equation)
If $A(t)$ is a time dependent family of gauge potentials on $\R^3$ and if 
$u := (d/dt)A(t)\Big|_{t=0}$ then $u$ is a $\kf$ valued 1-form on $\R^3$ and measures
the variation of $A(t)$ at $t =0$. For any function $f$ on $\A$ the derivative of $f$ at $A$ in the
 direction $u$ is given, in accordance with the chain rule, by 
\begin{align}
(\p_u f)(A) =  (d/dt) f(A(t))|_{t=0},          \label{py19}
\end{align}
where $A = A(0)$.

Suppose that for each t  the Poisson extension
of $A(t)$ is $\as_t(s)$.  The Poisson equation $\as_t''(s) = d_{\as_t(s)}^* \bs_t(s)$  
can be differentiated
with respect to $t$ at $t =0$ to find the variational equation associated to the variation $u$.
Define  $\ua(s) = (\p/\p t) \as_t(s)\Big|_{t=0}$. Then we find for $\ua(s)$ the 
variational equation
\begin{align}
\ua''(s) = d_{\as(s)}^* d_{\as(s)} \ua(s)  + [ \ua(s)\lrc \bs(s)].    \label{pe8}
\end{align}
Here $\as(s)\equiv \as_0(s)$ is the Poisson extension of $A \equiv A(0)$ and 
$\ua:[0, \infty) \rightarrow L^2(\R^3; \L^1\otimes \kf)$  is a function satisfying $\ua(0) =u$
in addition to \eref{pe8}.   
The last term in \eref{pe8} has been defined in the paragraph after Eq. \eref{pe6}.
In the following we will always assume that for a variation $u$ of interest to us there is a unique solution
to the variational equation  \eref{pe8} with initial value $u$ and such that $\ua(s)\rightarrow 0$ as $s\to \infty$.
We will refer to $\ua(s)$ as the {\it Poisson extension of $u$ along $\as$}.
}
\end{notation}

We can compute the derivative of the Poisson action as follows.

       \begin{lemma} \label{lempe2} 
       Let $u$ be a variation of a point  $A$ in $\A$,
        i.e. $u$ is a $\kf$ valued        1-form on $\R^3$. 
   If $\as(s)$ is the Poisson extension of $A$ then
\begin{align}
\p_u \P(A) = - 2(u, \as'(0))_{L^2(\R^3;\L^1\otimes \kf)}  .          \label{pe10}
\end{align}
\end{lemma}
       \begin{proof} Denote by $\ua(s)$ the Poisson extension of $u$.
Suppose that $A(t)$ is a curve with $A(0) = A$ and $(d/dt)A(t)|_{t=0} = u$ and that $\as_t(s)$ is the Poisson 
extension of $A(t)$. Then we have, at $t=0$, 
\begin{align*}
(d/dt) \| \as_t'(s)\|^2 = 2(\as'(s), \ua'(s))\  \text{and}\  (d/dt) \| \bs_t(s)\|^2 = 2(\bs(s), d_{\as(s)} \ua(s)).
\end{align*}
Therefore, using the Poisson equation \eref{pe6} in the third line below, we find 
\begin{align}
\p_u &\P(A) = 2 \int_0^\infty \Big( (\as'(s), \ua'(s)) + (\bs(s), d_{\as(s)} \ua(s)) \Big) ds    \label{pe11}\\
& = 2\int_0^\infty \(\Big\{(d/ds)(\as'(s), \ua(s))\Big\} - (\as''(s), \ua(s)) +  (d_\as(s)^*\bs(s), \ua(s))\) ds \notag\\
&= 2 (\ua(s), \as'(s))\Big|_0^\infty  \notag\\
&= -2(\ua(0), \as'(0)).   \notag
\end{align}
\end{proof}

\begin{remark}\label{remlong}{\rm (Longitudinal and  transverse) 
         The decomposition of an electromagnetic
field into longitudinal and transverse parts is not a gauge invariant decomposition. 
Nor  is there any  gauge invariant analog of this  decomposition for non-abelian gauge fields.
 But for the tangent spaces to the manifold 
 $\A$    there is a 
 well known gauge invariant  decomposition analogous to the Coulomb gauge. 
 The terminology, vertical and horizontal, for the differential analogs of longitudinal and transverse
   goes back to Ambrose and Singer \cite{AS}. This  will be reviewed here in our special context
 to help establish notation.
 }
\end{remark}

\begin{notation}\label{notvh}{\rm (Vertical and horizontal)  
If $\lambda:\R^3\rightarrow \kf$ is a reasonable function then $g(t,x): = \exp(t\lambda(x))$ defines a curve $t\rightarrow g(t)$
in the gauge group $\G$. Given a point $A \in \A$ the curve $t\mapsto A^{g(t)}$ is a curve in $\A$. 
Since, at $t=0$ one has $(d/dt)A^{g(t)} =[A, \lambda] + d\lambda$, the tangent to this curve in $\A$ at
$t =0$ is given by the 1-form $d_A \lambda := d\lambda + [A, \lambda]$. The {\it vertical subspace}
 of $T_A(\A)$
is by definition the set of tangent vectors to these curves at $A$. They consist therefore of $\kf$ valued 1-forms 
 $d_A \lambda$ with $\lambda$ running
 over the Lie algebra of $\G$. The  subspace  of $L^2$  orthogonal to the vertical subspace is
  called the {\it horizontal} subspace. Thus a $\kf $ valued 1-form $u$ on $\R^3$ is
 \begin{align}
& \text{ {\it vertical at $A$} if}\ u =d_A \lambda\ \ \text{for some scalar function $\lambda$}.\\
&  \text{ {\it horizontal at $A$} if}\ d_A^* u =0.
\end{align}
This is the standard terminology associated to the Coulomb connection for the $\G$ bundle $\A$,  \cite{Si2}.
      In this discussion we have made the usual identification of a tangent vector to $\A$ with a $\kf$ valued
      1-form on $\R^3$. The corresponding action of a gauge function on such a form is given by
      $u\mapsto u^g := g^{-1} u g$.
}
\end{notation}

     \begin{lemma} \label{lempy7}  
 If $f:\A \rightarrow \C$ is gauge invariant then
\begin{align}
\p_v f(A) = 0
\end{align}
for any vertical vector $v = d_A \lambda$.
\end{lemma}
       \begin{proof} Let $g(t, x)) = \exp(t \lambda(x))$. By assumption $f(A^{g(t)}) = f(A)$. 
But at $t =0$, $(d/dt) A^{g(t)} = v$. Therefore $(\p_vf)(A) =0$ by the chain rule.
\end{proof}

\begin{corollary} \label{corpy6}  
If $\as(\cdot)$ is a solution to the Yang-Mills-Poisson
 equation of finite action and $\as(0) = A$ 
then $\as'(0)$ is horizontal at $A$. That is, 
\begin{align}
d_A^* \as'(0) = 0.   \label{py20}
\end{align}
Moreover
\begin{align}
d_{\as(s)}^* \as'(s) = 0\ \ \  \forall s \ge 0.      \label{py25}
\end{align}
\end{corollary}
       \begin{proof}  Since $\P$ is gauge invariant Lemma \ref{lempy7}
       shows that $\p_v \P(A) = 0$ for any vertical 
       vector $v$, say $v= d_A \lambda$. But the derivative of $\P$ is given by \eref{pe10}, 
       from which we can conclude
       that $(d_A\lambda, \as'(0))_{L^2} =0$ for all $\lambda$. 
       That is, $(\lambda, d_A^* \as'(0))_{L^2} = 0$ for all $\lambda$.
       Hence \eref{py20} holds.      
       Now for any $\kf$ valued 2-form $ \w$ we have the identity $(d_A^*)^2 \w = [B\lrc \w]$, which follows
       by duality from the Bianchi identity $(d_A)^2 \phi = [B,\phi]$ for scalar functions $\phi$.
           Hence 
 \begin{align*}
 (d/ds)\Big(d_{\as(s)}^* \as'(s)\Big) &= d_{\as(s)}^*\as''(s) + [\as'(s)\lrc \as'(s)] \\
 &=  d_{\as(s)}^* d_{\as(s)}^* \bs(s) + 0\\
 &= [\bs(s)\lrc \bs(s)] \\
 &=0.
 \end{align*}
 Therefore $d_{\as(s)}^* \as'(s)$ is constant in $s$. 
 Since, by \eref{py20},  it is zero at $s=0$ it is identically zero. 
\end{proof}

     \begin{theorem} \label{thm2derivs}   $($Second derivatives of the Poisson action$)$ 
Suppose that $A$ has finite Poisson action and that 
$u$ and $v$ are $\kf$ valued 1-forms on $\R^3$. 
Denote by $\as$ the Poisson extension of $A$ and by $ \ua, \va$ the  respective Poisson extensions
of $u$ and $v$ along $\as$. Then
\begin{align}
(1/2) \p_v\p_u \P(A) &= \int_0^\infty \Big\{ \(\ua'(s), \va'(s)\)_{L^2} + \(d_{\as(s)} \ua(s), d_{\as(s)} \va(s)\)_{L^2}  \notag\\
              &\qquad     + \([\ua(s) \wedge \va(s)], \bs(s)\)_{L^2}\Big\} ds  \label{py205}
\end{align}
and also
\begin{align}
- (1/2) \p_v\p_u \P(A) =  \(u, \va'(0)\) = \(v, \ua'(0)\)   .                 \label{py207}
\end{align}
In particular,
\begin{align}
(1/2) \p_u^2 \P(A) = \int_0^\infty\Big\{ \|\ua'(s)\|_{L^2}^2 + \|d_{\as(s)} \ua(s)\|_2^2 
+\([\ua(s)\wedge \ua(s)], \bs(s)\)\Big\}  ds                            \label{py208}
 \end{align}
 and also
 \begin{align}
  \p_u^2 \P(A) = -2(u, \ua'(0)) = -(d/ds) \| \ua(s)\|_{L^2}^2 \Big|_{s=0} .  \label{py209}
  \end{align}  
\end{theorem}
      \begin{proof} Differentiate \eref{pe10} with respect to $A$ in the direction $v$ to find
      \begin{align}
 - (1/2) \p_v\p_u \P(A) =   \(u, \va'(0)\)_{L^2}  .    \label{py210}
 \end{align}
  Using the variational equation \eref{pe8} we get 
 \begin{align}
   &(d/ds)(\ua(s), \va'(s))      \notag\\
   &= (\ua(s), \va''(s)) + (\ua'(s), \va'(s)) \notag\\
   &= (\ua(s), d_\as^* d_\a \va(s) + [\va(s) \lrc \bs(s)])   + (\ua'(s), \va'(s))  \notag\\
   &=(d_\as \ua(s), d_\as \va(s))  + (\ua(s),[\va(s) \lrc \bs(s)]) + (\ua'(s), \va'(s))\notag\\
   &= (d_\as \ua(s), d_\as \va(s)) +   ([\va(s)\wedge \ua(s)], \bs(s))  + (\ua'(s), \va'(s)). \label{py216}
      \end{align}
    Assuming that   $(\ua(s), \va'(s))$ goes to zero
    as $s\to \infty$ we may integrate this  identity over $[0, \infty)$ to find, using  $[u\wedge v ] = [ v\wedge u]$,
    \begin{align}
   & -(\ua(0), \va'(0))        \label{py217}\\
    &\ \ \ \  \ \ \ \  =   \int_0^\infty \Big\{ (d_\as \ua(s), d_\as \va(s)) 
    +   ([\ua(s)\wedge \va(s)], \bs(s))  + (\ua'(s), \va'(s))\Big\} ds.  \notag  
    \end{align}
    \eref{py205} and \eref{py207} now follow from \eref{py210}, \eref{py217} and the
     symmetry in $u$ and $v$ in \eref{py217}.      
         Set $v =u$ in \eref{py205} and \eref{py207} to find \eref{py208} and \eref{py209}.
      \end{proof}

\begin{remark} {\rm (Expansion of $\P$ near zero.) 
 The behavior of $\P(A)$ for small $A$ is given by
\begin{align}
\P(u) = (|C| u, u)_{L^2(\R^3; \L^1\otimes \kf)} + O(u^3) ,      \label{py230}
\end{align}
where $u$ is a small variation of $A$ at $A=0$. 
   To see this observe that the Poisson extension of $A=0$ is $\as(s) =0$. The variational equation
   \eref{pe8} is therefore $\ua''(s) = d^*d\ua(s)$, which is the Maxwell-Poisson equation for  the  
   $\kf$ valued 1-form $\ua(s)$. The solution, as in the real valued case, \eref{mp30},  is $\ua(s) = e^{-s|C|} u$,  where
   $C$ is the curl acting on each $\kf$ component of $u$. Hence $\ua'(0) = - |C|u$. We find therefore,
   $\P(0) =0$, $(\p_u \P)(0) =  - 2(u, \as'(0)) =0$ and $(1/2)(\p_u^2 \P)(0) = (|C| u, u)$, from which 
   \eref{py230} follows.
}
\end{remark}

     \begin{remark} {\rm It may be useful to observe that $\as'(s)$ is itself a variational field along $\as$,
 and if $\ua$ is also  a variational  field along $\as$ then 
      \begin{align}
 (\ua(0), \as''(0)) = (\ua'(0), \as'(0)).       \label{pe18}
 \end{align}
           \begin{proof}
  Let $\va(s) = \as'(s)$. Then
     \begin{align*}
     \va''(s) &= (d/ds) \as''(s)\\
     &=(d/ds) d_\as^*\bs \\
     &= d_\as^* d_\as \a'(s) + [\as'(s) \lrc \bs(s)]\\
     &= d_\as^* d_\as \va(s) + [\va(s) \lrc \bs(s)].
     \end{align*}
     Therefore $\as'$ is a variational field along $\as$. Upon taking $\va(s) = \as'(s)$ in  \eref{py207} we find
     \eref{pe18}. 
          \end{proof}
     }
\end{remark}

\subsection{The Riemannian metrics on $\A$ and $\cg$}  \label{secrm}
We want to define a Riemannian metric on $\A$ which is gauge invariant, captures  a
non-linear version of the electromagnetic norm identity \eref{mp35}, and, unlike \eref{mp35}, does not require
the Coulomb gauge. We aim to do this by using the first variation of the Poisson action 
to define a norm on horizontal vectors and then using 
an explicitly defined  norm on vertical vectors that's commensurate 
with our choice of gauge group, $\G$.

        \begin{notation}\label{notrma} {\rm (Riemannian metric on $\A$) 
 Suppose that $A$ is in $\A$ and that  $\as$ is its Poisson extension. 
 If $u$ is any $\kf$ valued 1-form on $\R^3$ and  $\ua$ is  its Poisson extension along $\as$ 
define
\begin{align}
_A\|u\|^2 = \int_0^\infty \Big( \|\ua'(s)\|_2^2 + \| d_{\as(s)} \ua(s)\|_2^2 \Big) ds.\ \  \label{pe13}
\end{align}
Suppose that $w$ is a $\kf$ valued 1-form on $\R^3$ which decomposes into horizontal and vertical
components:
\begin{align}
 w = u + v,      \ \ \ \  d_A^*u =0,\ \ \ \ v=d_A\lambda .  \notag
 \end{align}
 Define 
\begin{align}
\|w\|_A^2  = \| u + v\|_A^2 =\   _A\|u\|^2 + \| (-\Delta_A)^{1/4} v\|_{L^2(\R^3)}^2,  \label{py32}
\end{align}
where $\Delta_A = \sum_{j=1}^3 (\p_j^A)^2$ and $\p_j^A = \p_j + [A_j, \cdot ]$.
We define the horizontal subspace at $A$ by  
 \begin{align}
  u  \in H_A\ \ \ \text{if}\ \  
  \begin{cases} &d_A^* u =0\ \ \ \text{and} \\
                         &\|u\|_A < \infty.\ \      \label{pe9}
   \end{cases}
  \end{align} 
 It seems likely that the norm $\|u\|_A$ is non-degenerate on $H_A$ and we will assume this in the  following.
But the integral in \eref{pe13} is meaningful even if $u$ is not horizontal at $A$. 
However $_A\|u\|$ is not a norm since it can be degenerate
on non-horizontal elements, as will be shown in the next example. 
For this reason we have  added  the second term  in \eref{py32}. 
We will see that it  is commensurate with our choice of $\G$, defined in \eref{py17}.  

In the case of electromagnetism, where $K = U(1)$, \eref{mp35} shows that $\P(A) = \|A\|_{H_{1/2}}^2$
 when $A$ is in Coulomb gauge. Identifying, as usual, the tangent space to $\A$ with $\A$ itself in
 this linear theory, we see that the norm $\|u\|_A$ in \eref{py32} reduces to
  the $H_{1/2}$ norm of electromagnetism
 when $u $ is horizontal. One should think of $\|u\|_A$ as an ``$H_{1/2}$''-like
  norm on $H_A$ in the non-abelian case.
}
\end{notation}

\begin{example}\label{exdegen}{\rm   (Degeneracy) 
If $A=0$ then,  for any $\kf$ valued function $\lambda$ on $\R^3$, 
$d\lambda$ is a vertical vector at $A$ and we have
\begin{align}
_A\| d\lambda \|   =0.                                               \label{pe13a}
\end{align}
This is easily seen because the Poisson extension of $A:=0$ is $\as(s) \equiv 0$.
Therefore if we define $\ua(s) = d\lambda$ for all $s \ge 0$ then
$\ua'(s) = \ua''(s) =0$, while  $d^*d\ua(s) = d^*d^2\lambda =0$ and $\bs(s) = d^2 \lambda = 0$,
from which we see that the variational equation \eref{pe8} is satisfied.
 Since also  $\ua(0) = d\lambda$, $\ua(s)$ is the Poisson extension of  $d\lambda$. But then 
 the right hand side of \eref{pe13} is zero.
 }
 \end{example}

\begin{lemma}\label{lemgiv}  $($Gauge invariance of the metric$)$
Let $w \in T_A(\A)$ and let $g\in \G$.
Suppose that $w = u +v$ with   $u$ horizontal at $A$ and $v$ vertical at $A$.
 Then $w^g = u^g + v^g$ with $u^g$ horizontal at $A^g$ and $v^g$ vertical at $A^g$. Moreover
\begin{align}
\|w^g\|_{A^g} = \|w\|_A     \label{pe13b}
\end{align}
\end{lemma}
     \begin{proof}  Horizontal and vertical are preserved under gauge transformations because
     $d_{A^g}^* u^g = (d_A^* u)^g$ and $(d_A\lambda)^g =d_{A^g}\lambda^g  $.
 In view of the definition \eref{py32} it suffices to show that \eref{pe13b} holds for $u$ and $v$ separately.

Concerning the horizontal component, let $\as$ be the Poisson extension
 of $A$ and let $\ua$ be the Poisson extension
of $u$ along $\as$. Then $\as^g$ is the Poisson extension of $A^g$, as we saw in  Section \eref{secpacs},
and $\ua^g(s) :=  g^{-1} \ua(s) g$ is the Poisson extension of $u^g$, as one sees by a similar argument.
Since $\| (\ua^g)'(s)\|_2 = \| g^{-1} \ua'(s) g \|_2 = \| \ua'(s)\|_2$ and
 $ d_{\as^g(s)} \ua^g(s) = g^{-1} d_{\as(s)} \ua(s) ) g$, \eref{pe13b} follows from \eref{pe13} for the horizontal
 component of $w$.

 Concerning the vertical component, observe that the gauge covariant Laplacian $\Delta_A$ commutes
 with gauge transform in the sense that $\Delta_{A^g} v^g = g^{-1} (\Delta_A v) g$ for any $\kf$ valued 1-form
 $v$. Since $Ad\ g^{-1}$ is unitary in $L^2(\R^3;\L^1\otimes \kf)$, it follows  from the spectral theorem
 that $(-\Delta_{A^g})^{1/4} v^g = g^{-1} ((-\Delta_A)^{1/4} v) g$. Therefore
 $\| (-\Delta_{A^g})^{1/4} v^g\|_{L^2} = \| (-\Delta_A)^{1/4} v\|_{L^2}$.
  \end{proof}

\begin{corollary} \label{corpy5}
\begin{align}
|\p_u \P(A)| \le 2 \sqrt{\P(A)} \| u\|_A,\ \   u \in H_A(\A)   \label{pe14} \\
|(u, \as'(0))| \le   \sqrt{\P(A)} \| u\|_A, \ \   u \in H_A(\A).    \label{pe15}
\end{align}
\end{corollary}
\begin{proof}
From \eref{pe11} and the Schwarz inequality we see that
\begin{align}
|\p_u &\P(A)| \le 2 \sqrt{\P(A)} 
\sqrt{\int_0^\infty \Big( \|\ua'(s)\|_2^2 + \| d_{\as(s)} \ua(s)\|_2^2\Big) ds },         \label{pe12}
\end{align}
which is \eref{pe14}. \eref{pe15} follows from \eref{pe10}.
\end{proof}

Note: Both inequalities hold for all $u \in T_A$ because $\p_v \P(A) =0$ for vertical $v$, by Lemma \ref{lempy7}.

\begin{notation}\label{notrmc} {\rm (Riemannian metric on $\cg$)  
Since $\G$ acts isometrically on $\A$ it induces a metric on $\cg : = \A/\G$, 
which makes the projection into a submersion of Riemannian manifolds.
We take this induced metric as the Riemannian metric on $\cg$.
}
\end{notation}

\begin{remark}\label{remkindyn}{\rm  (Kinematic and dynamic norms) 
A point $q$ in configuration space is a $\G$ orbit, $q:= A\G$, through a point in $\A$.
 The tangent space to such an orbit  is the set of  vertical vectors $d_A\lambda$ at $A$ and 
the tangent space to $\cg$ at the orbit  $q$  
 can be identified with the set of horizontal  vectors at $A$.  
    Since the norm defined in \eref{py32}  
    is gauge invariant in the sense of \eref{pe13b}, it 
  descends to a norm on $T_q(\cg)$. 
        But the simple $L^2$ norm $ \| u\|_{L^2(\R^3; \L^1\otimes \kf)}$ on $H_A$
  is also gauge invariant and  also descends to a norm on $T_q(\cg)$.
      Unlike the first  norm  the $L^2$ norm does not depend in any way on
       the dynamics, which was used e.g. in \eref{pe6}. The $L^2$ norm 
       is entirely kinematic and should appropriately        be referred to as the
     {\it kinematic} Riemannian metric  for configuration space, in contrast to the {\it dynamic} Riemannian
     metric \eref{py32}.  
           We already saw in the electromagnetic case that the norm on phase space \eref{n6},
            induced from the dynamic
            norm on configuration space by the duality set up by the kinematic norm, gives
              the unique Lorentz   invariant norm. Both the kinematic and dynamic norms thereby 
              appear naturally in considerations of phase space. The kinematic metric is needed to
               construct the Coulomb connection. See e.g. \cite{Si1,Si2,NR}.              
              Other   Sobolev norms have been used in some works  to ensure
               that only irreducible connections need be considered, thereby ensuring in turn 
            that Green functions exist. See e.g. \cite{MV81,AM}. 
            These strong norms enter in a technical way only.
  }
\end{remark}

\subsection{Phase space }               \label{secpsy}

The phase space for the gauge field is $T^*(\cg)$, as usual in classical mechanics. 
The natural Riemannian metric on $\cg$ to use here is the one defined in Notation \ref{notrmc},
which gives the Lorentz invariant norm in the electromagnetic case. 
An  electric field   is a member  of $T_A^*(\A)$ for some point $A \in \A$. 
In the electromagnetic case the space $T_A^*(\A)$ can be taken to be independent of $A$ because
$\A$ is a linear space. The electric field can therefore be, and customarily is,  taken to be
 independent of $A$. But not in a non-abelian gauge theory.

If the electric field $E$ annihilates all vertical vectors at $A$, that is,
 $\<d_A \lambda, E\>_{L^2(\R^3)}=0$ for 
all $\kf$ valued scalar functions $\lambda$,  then $d_A^* E =0$. Otherwise one has $d_A^*E = - 4\pi \rho$
for some $\kf$ valued charge distribution $\rho$. One can see already from this that charge is connected
to vertical vectors by some duality and therefore also to the dual space of the Lie algebra of the gauge group.
 At an informal level this statement
just reflects the well recognized 
 fact that conservation of  charge is a consequence of gauge invariance.  
 See e.g. \cite[pages 210-211]{Sch}.
  But the charge  will not show up in the structure of $T^*(\cg)$ because the projection 
  of $E$ into $T^*(\cg)$ depends on evaluation of $E$ on horizontal vectors only. 
 To incorporate 
 charge, from either external sources or additional fields, one must make further use of $\A$ and the
  gauge group $\G$,  beyond their roles in forming the configuration space $\A/ \G$. 
  The incorporation of charge into our phase space framework will 
 be addressed quantitatively in a future work. But in the present paper we are only concerned with 
 helicity and therefore we will always take the charge to be zero.

  As already noted in  Remark \ref{rmkpm}, only the decomposition of configuration space 
  must be specified in order to define helicity, whether in the classical or quantum theory,
  because the associated decomposition of $T^*(\cg)$ is automatically determined
   by this decomposition in the classical theory  
  and is not needed in the quantum theory. 
  In the remainder of this paper 
  we will be concerned 
  only with decomposition  of the non-abelian configuration
              space $\cg$ into submanifolds $\cg_\pm$ analogous to those of the electromagnetic case.
   We limit further discussion of phase space to the following remarks.
   
\begin{remark}\label{rmkLinvce} {\rm   (Lorentz invariance)
   In the case of electromagnetism we saw 
       how the phase space  norm gives the Bargmann-Wigner
       Lorentz invariant norm.  (cf. \eref{pw36c}.) 
       For non-abelian fields the configuration space $\cg$ is not linear and
       one must replace what was previously a norm on $\cg \oplus \cg^*$ (namely \eref{pw36c}) by
       a Riemannian metric on $T^*(\cg)$. This space automatically inherits a Riemannian
        metric from the Riemannian metric on $\cg$, defined in Notation \ref{notrmc}. 
        It remains to be seen whether this Riemannian
        metric is invariant under the action of the Lorentz group. Invariance under spatial rotations and
        spatial translations is clear. But invariance under time translation is not.           
         It is not at all clear  that the Riemannian metric on $T^*(\cg)$ needs to be Lorentz invariant
          in order for    the configuration space 
        itself to play its usual role in the quantized theory,
         where quantum time evolution is not so simply related to the non-linear classical evolution.
}
\end{remark}

\begin{remark}\label{rmkHSC} 
{\rm (History, Symplectic form, Complex structure) Candidates for the phase space
 associated to linear and non-linear wave equations
have been investigated from many  different viewpoints. One expects a symplectic structure
on the phase space and in some cases one can hope also for a Riemannian structure
 on this infinite dimensional manifold.
 In field theory the symplectic structure  is an infinite dimensional analog of the
 canonical 2-form   $\sum_{j=1}^n dp_j\wedge dq_j$. There is a large literature devoted to 
 the role of the symplectic structure in quantization of the field theory.
 See e.g.  
 \cite{Seg41,SegB1,Seg117, Seg123, Seg126,Seg128,Pan2,Woodhouse,Garc1,Garc2, Garc3, Garc4, Garc7}
 for samples. The phase space is usually not chosen in a quantitative way, however, as we have done above.
 But the references \cite{Seg117, Seg123, Seg126,Pan2} do  discuss the symplectic form in combination
 with a Riemannian metric on the phase space. See especially \cite{Seg126} for a discussion
  of phase space for  Yang-Mills fields and a proposed quantitative candidate for it.
 In the context of the Yang-Mills theory the symplectic form is given by the 2-form,
\begin{align}
\w\< X_1, X_2\> = \int_{\R^3}\Big( \<u_1(x), v_2(x)\> - \< u_2(x), v_1(x)\> \Big) d^3x
\end{align}
where  $X_j := u_j, v_j$, $j = 1,2$ are two elements of $T_{A,E} T^*(\A)$ and 
$\< u(x), v(x)\> =  \< u(x), v(x)\>_{\L^1\otimes \kf}$. 
If $A$ evolves by the hyperbolic Yang-Mills equation and $X_j(t)$ evolves by the induced flow then
it is known that
\begin{align}
(d/dt)\w\< X_1(t), X_2(t)\> = 0. \notag
\end{align}
Moreover  $\w$ is gauge invariant, Lorentz invariant and descends to a non-degenerate
 symplectic form on $T^*(\cg)$.  See e.g. \cite{Seg41}-\cite{SegB1} and especially \cite{Seg126}
  for early discussions of these structures.
The 2-form $\w$ is continuous in the $u_j$ and $v_k$ because they are in dual spaces.
But our particular choice of the metric on configuration space gives the Sobolev indices
  $\pm 1/2$ for these spaces, which  match fortuitously with the
 expectation that the electric field
should be one derivative less regular than the potential.

The complex structure on the linear phase space for electromagnetism, used in the proof of Theorem
\ref{thmhel1}, has a natural analog in the non-abelian theory. If $u$ varies $A$ and $v$ varies $E$
then the map $j:\{u, v\} \mapsto \{\hat v, - \hat u\}$, where 
$ T_{A,E} \hat\longleftrightarrow T_{A,E}^*$ is the natural
isometry, generalizes the electromagnetic complex structure given in the proof of Theorem \ref{thmhel1}.
 }
\end{remark}

\section{Decomposition of Yang-Mills configuration space by (anti-) self-duality}  \label{secdecomp}

\subsection{Definition of $\cg_\pm$} \label{seccpm}
The four dimensional curvature of the 1-form $\as$ on the half space $\R^3\times [0,\infty)$ is given
by \eref{py4}. Its Euclidean dual 
is  given by $*_e F = ds\wedge *\bs(s) +  *\as'(s)$, where $*$ is the three dimensional Hodge star operator.
As in the discussion in Section \ref{secsdm} the conditions for Euclidean self-duality
 or anti-self-duality,  $*_e F = \pm F$, of the curvature of $\as(\cdot)$ can be written in terms of $\as$ as follows.

\begin{definition}{\rm A function $\as: (0, \infty) \rightarrow L^2(\R^3; \L^1\otimes \kf)$ is {\it self-dual} if 
\beq
\as'(s) = *\bs(s),\ \ \ \ \ 0 < s < \infty.          \notag
\eeq
It  is {\it anti-self-dual} if 
\beq
\as'(s) = -*\bs(s),\ \ \ \ \ 0 < s < \infty.         \notag
\eeq
}
\end{definition}

The following lemma restates for temporal gauge a well known fact.  
      \begin{lemma} \label{lempe15} If $\as(\cdot)$ is self-dual or anti-self dual then it satisfies
 the Yang-Mills-Poisson equation \eref{pe6}.
\end{lemma}
      \begin{proof} If $\as'(s) = *\bs(s) \epsilon$ with $\epsilon = \pm1$ then 
\begin{align*}
\as'' &= (d/ds) *\bs(s)   \epsilon\\
&= *d_{\as(s)} \as'(s)   \epsilon\\
&=  *d_{\as(s)} *\bs(s) \epsilon^2\\
&= d_{\as(s)}^* \bs(s)
\end{align*}
since $*d_{\as(s)} * = d_{\as(s)}^*$ when acting on two forms.
\end{proof}

\begin{lemma}\label{lempe16}Let $g:\R^3\rightarrow K$ be a gauge function. 
If $\as$ is $($anti-$)$self-dual then so is $\as^g$. 
\end{lemma}
\begin{proof} Since the three dimensional curvature of $\as^g(s)$ is $g^{-1} \bs(s) g$ we
 have,  for $\epsilon = \pm 1$,
\begin{align}
 (d/ds) \as^g(s) = g^{-1} \as'(s) g = g^{-1} * \bs(s)\epsilon\, g
           = * g^{-1} \bs(s) g\, \epsilon=  *Curv\ \as^g(s)\,  \epsilon.          \notag           
 \end{align}
 where Curv means the three dimensional curvature. 
 \end{proof}

\begin{notation}\label{notasd}{\rm ($\Ra_{\pm}$, $\A_{\pm}$, $\cg_{\pm}$)  Denote by $\Ra_+$ the
 set of anti-self-dual solutions
in $\Ra$  and by $\Ra_-$ the set of self-dual solutions in $\Ra$.  $\Ra_+$
and $\Ra_-$ are each gauge invariant sets by Lemma  \ref{lempe16}.  Let
\begin{align}
\A_+  &:= \{ \as(0): \as \in \Ra_+\}, \\
\A_- &:= \{ \as(0): \as \in \Ra_-\}.
\end{align}
$\A_\pm$ are gauge invariant sets because if $\as$ is the Poisson extension of $A$
then $ \as^g$ is the  Poisson extension of  $A^g$.
 We define submanifolds of configuration space by 
\begin{align}
\cg_+ = \A_+ /  \G,\ \ \  \cg_- = \A_-  /   \G
\end{align}
}
\end{notation}  
The definitions of $\cg_\pm$ given in  Notation \ref{notasd}  
in the non-abelian case are precisely analogous to those for the electromagnetic field and reduce to them when
$K = U(1)$.
  We saw that   in the electromagnetic case these manifolds control helicity in the sense that
a solution to Maxwell's equations with initial data in $T^*(\cg_+)$  has only positive helicity plane waves
in its plane wave decomposition. Similarly for $T^*(\cg_-)$. 
In the next section and in  \cite{G76}  we will provide further justification 
  for the interpretation of the manifolds $\cg_\pm$ as characterizers of helicity  in the non-abelian case .

\begin{remark}\label{remclosed}{\rm 
 In the Riemannian metric on $\A$ defined in Section \ref{secrm}
the subsets $\A_\pm$ are closed in $\A$. Since $\G$ acts as isometries on $\A$ the quotient spaces
$\cg_\pm$ are also closed submanifolds  of $\cg$. 
}
\end{remark}

\subsection{Flows of the helicity vector fields and positivity of $curl_A$} \label{secflows}
 We will give support in this section for the interpretation of the
  submanifolds $\cg_\pm$   as    the 
 non-abelian analogs of the electromagnetic helicity subspaces.  
 For non-abelian fields we have used (anti-)self duality to define the submanifolds $\cg_\pm$
  in Section \ref{seccpm}. 
  Now   we are going to show how this decomposition of configuration space relates to the signature of the 
 covariant curl. We will also show that the orthogonal direct sum \eref{mp5}, that holds in electromagnetism,
 has an analogous product decomposition in the non-abelian case.

  The transition  from the electromagnetic linear space theory to the Yang-Mills 
  non-linear theory   will be    guided by
  the following observations,   which will gradually shift the paradigm from
   functional analysis over the electromagnetic 
  Hilbert space  $\cg$  to infinite dimensional differential geometry over the 
  analogous non-abelian configuration space. Denote again by $C$ the operator curl.
  
  a. In electromagnetism the positive spectral subspace of $C$  
  is the null space of the operator $|C| - C$.
  
  b. The function 
  \beq
  h_+(A):= -(|C| - C)A       \label{py140}
  \eeq
   can be interpreted as a vector field  on the electromagnetic Hilbert space $\cg$.
  As such, it has a flow  $\exp(th_+)$. It will be shown that the flow  converges
   as $t\to \infty$ to the orthogonal
     projection $P_+$ of $\cg$ onto    $\cg_+ $

c. The operator $|C|$, whose definition as $\sqrt{C^2}$ has a functional analytic character, 
can be replaced in  electromagnetism 
by use of the Maxwell-Poisson semigroup, which has a PDE character. Indeed we see from \eref{mp30}
that  $|C|A =  -\as'(0)$, where $\as(s)$ is the solution to the Maxwell-Poisson equation with initial value $A$.

d. The second term in the vector field $h_+(A)$ in \eref{py140} is just the magnetic
 field  $curl\, A$   in its usual  representation as a 1-form (which is better written as $*B$ with $B=dA$).
  Thus the vector field $h_+$  can be specified  
 in electromagnetism   by the expression
 \beq
 h_+(A) = \as'(0) + *B,  \label{py141}
 \eeq
  thereby bypassing 
 the  functional analytic construction  
 $|C| =\sqrt{C^2}$.
 
e. 
(Decomposition of configuration space) In the non-abelian theory we want to avoid making any
 gauge choice, such as the Coulomb gauge. The expression \eref{py141} defines  a  vector field on $\A$ 
 in the non-abelian case. 
We will show that it flows exactly onto $\A_+$ as $t\to\infty$, analogous to  the abelian case, thereby 
 producing a non-linear map $P_+:\A \rightarrow \A_+$. Similarly for $\A_-$. 
 Upon quotienting by the gauge group, we arrive at a decomposition
 $\cg = \cg_+ \times \cg_-$. The one-to-one ness of this map needs to be verified, however.

 f. (Positivity of the covariant curl.) 
 In electromagnetism one has \linebreak
  $(curl A, A)_{L^2(\R^3)} \ge 0$  for all $A \in \cg_+$ by   \eref{he202} and Theorem \ref{thmhel1}.
 In the non-abelian case $curl_A$, the gauge covariant curl,  does not act on connection
  forms $A$ but on tangent    vectors to $\A$. We will  show that   
 $(curl_A u, u)_{L^2(\R^3)} \ge 0$ for all $A \in \A_+$  
 and for all tangent vectors $u$ to $\A_+$ 
 at $A$.  The form of the last statement   reflects the 
 change often  needed 
when replacing a linear space by a non-linear manifold, 
 where  one can no longer identify 
the  manifold with its tangent spaces.

       \begin{definition}{\rm For a point $A$ in $\A$ denote by $B$ its
         curvature: $B = dA + A\wedge A$ and by $\as(s)$  the  Poisson extension of $A$. Define 
\begin{align}
h_+(A) &= \as'(0) +   *B     \label{pe150}\\
h_-(A) &= \as'(0) -   *B. \label{pe151}
\end{align}
Note that $h_\pm(A)$ are 1-forms and therefore can be identified as tangent vectors to $\A$ at $A$. 
$h_\pm$ therefore can and will be identified as  vector fields on $\A$. Moreover they are both gauge covariant:
\begin{align}
h_{\pm} (A^g) = g^{-1} h_{\pm}(A) g.               \notag
\end{align}
}
\end{definition}

      \begin{theorem}\label{thmf1} \ 
      
           a.  The vector   fields $h_\pm(A)$ are   horizontal at $A$ and mutually orthogonal 
            in $L^2(\R^3)$ for each point $A \in \A$. 
          
     \bigskip
     b.  $h_\pm$ is zero exactly on $\A_\pm$. 

    \bigskip

     c.     The flows $\exp{th_\pm}$ generated by the vector fields $h_\pm$ for $t \ge 0$ satisfy
 \begin{align}
P_+(A) & := \lim_{t\to \infty} (\exp th_+)A\  \ \text{lies in}\ \  \A_+.          \label{py160}\\
P_-(A) &:= \lim_{t\to \infty} (\exp th_-)A\  \ \text{lies in}\ \  \A_-.            \label{py161}\\
P_\pm(A) &= A \ \text{if}\ \ A\in \A_\pm.                                                                \label{py162}
\end{align}

\bigskip

 d. Suppose that the Poisson action is convex at $A$. That is, 
 \beq
 \p_u^2\P(A) \ge 0\     \text{for all}\ u.
 \eeq 
 If $A\in \A_+$ and $u$ is tangential to $\A_+$ at $A$ then
 \begin{align}
 (curl_A u, u)_{L^2} \ge 0.           \label{py163}
 \end{align}
 
 \noindent
 If $A\in \A_-$ and $u$ is tangential to $\A_-$ at $A$ then
 \begin{align}
 (curl_A u, u)_{L^2} \le 0.        \label{py164}
  \end{align}
\end{theorem}

The proof depends on the following lemmas.
\begin{lemma}    \label{lemf2}
\begin{align}
&a.)\ \ \ d_A^*\ h_\pm(A) =0 \ \ \              \forall\ \ A \in \A.                                        \label{f9} \\
&b.)\ \ \ (h_+(A), h_-(A))_{L^2(\R^3; \L^1 \otimes \kf)} = 0\ \ \ \forall\ \ A \in \A.             \label{f11} \\
&c.)\ \ \ h_+(A) =0\ \ \text{if and only if}\  A \in \A_+\ \ \text{and} \label{f13}\\
&\ \ \ \ \ \ h_-(A) = 0\ \ \text{if and  only if}\  A \in \A_-. \notag
\end{align}
\end{lemma}
 \begin{proof}  From \eref{py20} we see that $d_A^* \as'(0) = 0$. Moreover 
 $d_A^* (*B) = -(*d_A*)* B = -*d_A B =0$
 by the Bianchi identity.  We have used here the identity $d_A^* = -* d_A *$ when acting on 1-forms.
Both terms in $h_\pm(A)$ are therefore horizontal at $A$, which proves \eref{f9}.

To prove \eref{f11}  we will first establish the identity
 \begin{align}
 \|\as'(s)\|_2^2 = \| \bs(s)\|_2^2\ \ \ \forall \ \ s \ge 0.               \label{f10}
 \end{align}
 The computation
 \begin{align*}
(1/2) (d/&ds)\Big(\|\as'(s)\|_2^2 - \| \bs(s)\|_2^2) \\ &=(\as''(s), \as'(s)) - (\bs'(s), \bs(s)) \\
&=(d_{\as(s)}^* \bs(s), \as'(s)) - (d_{\as(s)} \as'(s), \bs(s)) \\
&=(\bs(s), d_{\as(s)} \as'(s)) -    (d_{\as(s)} \as'(s), \bs(s)) \\
&=0
\end{align*}
shows that   $\|\a'(s)\|_2^2 - \| \b(s)\|_2^2$ is constant in $s$. Since $\P(A) < \infty$ it follows that
the constant is zero. This proves \eref{f10}.
           Using \eref{f10} we find  
\begin{align*}
(h_+(A), h_-(A))_2 &= (\as'(0), \as'(0))_2 - (*\bs(0), *\bs(0))_2 \\
&= \|\as'(0)\|_2^2 - \|\bs(0)\|_2^2 \\
&=0.
\end{align*}
This proves \eref{f11}.

 For the proof of \eref{f13} observe that if $A \in \A_+$ and $\as$ is its Poisson extension then
 $\as'(s) +  *\bs(s) =0$ for all $s \ge 0$ and in particular $h_+(A) = \as'(0)+   *\bs(0) =0$. 
 Similarly, if $A \in \A_-$ then $h_-(A) =0$.

 To prove the converses let   $\alpha(s) = \as'(s)+*\bs(s)$. 
 Then
            \begin{align*}
 (d/ds) \alpha(s) &= \as''(s) +       * d_\as \as'(s) \\
 &= d_\as^*\bs(s) +  *d_\as \as'(s) \\
 &= *d_\as(*\bs(s) +  \as'(s) )\\
 &= *d_{\as(s)} \alpha(s).
 \end{align*}
 This is  a linear homogeneous first order differential equation in $\alpha(s)$. Therefore if $\alpha(0) = 0$ then
  $\alpha(s) \equiv 0$. 
   So if $h_+(A) =0$ then   $\as'(s)+   *\bs(s) =0$
   for all $s \ge 0$. Therefore $\as$ is anti-self-dual and $A \in \A_+$. 
  Similarly if $h_-(A) =0$ then  the function $\beta(s)\equiv \as'(s) -   *\bs(s)$ is zero at $s=0$ and,
  since it satisfies the differential equation $ (d/ds) \beta(s) = -   *d_{\as(s)} \beta(s)$, 
  it is identically zero.
  So $\as$ is self-dual and $A \in \A_-$.
\end{proof}

        \begin{lemma}\label{lemf5} For $t \ge 0$, $\P(\exp(th_\pm)A)$ is a non-increasing function of $t$ for
          each  $A \in \A$  and for each of the two flows.
\end{lemma}
            \begin{proof} It suffices to compute the derivative of $\P(  \exp(th_\pm) A)$ at $t =0$ because
$\exp(th_{\pm}) $ are semigroups. At $t=0$ we have, by \eref{pe10},
       \begin{align*}
 (1/2)(d/dt)\P(\exp(th_+) A)\Big|_{t=0} &= (1/2)\p_u \P(A)\Big|_{u = h_+(A) } \\
 &=-(u, \as'(0))\Big|_{u = h_+(A)} \\
 &= -(\as'(0) +  *B, \as'(0)) \\
 &= - \|\as'(0)\|_2^2 -  (*\bs(0), \as'(0)).
 \end{align*}
 But from  the Schwarz inequality and \eref{f10} we find
 \begin{align*}
 |(*\bs(0), \as'(0))| &\le \|\bs(0)\|_2 \|\as'(0)\|_2 \\
&\le \|\as'(0)\|_2^2.
 \end{align*}
 Hence
 \begin{align}
 (d/dt)\P(\exp(th_+) A)\Big|_{t=0} \le 0. \notag
 \end{align}
 A similar proof applies to $\exp(th_-)$. 
 \end{proof}

\begin{lemma} \label{lemf6}  
Let $A \in \A$.

 If $h_+(A)\P(A) =0$ then $h_+(A) =0$.
 
If $h_-(A)\P(A) =0$ then $h_-(A) =0$.
\end{lemma}
  \begin{proof} The hypotheses are short for  $\p_u\P(A)\Big|_{u=h_\pm (A)} =0$. 
  Using \eref{pe10} we see that 
  \begin{align}
  0 =(1/2) h_+(A)\P(A)  &=- (h_+(A), \as'(0)) \\
  &= -(\as'(0)+  *B, \as'(0))\\
  &=-\|\as'(0)\|_2^2   -(*\bs(0), \as'(0)).
  \end{align}
  Therefore
  \beq
 -(*\bs(0), \as'(0)) = \| \as'(0)\|_2^2
  \eeq
  But $\| -*\bs(0)\|_2 = \|\as'(0)\|_2$ by \eref{f10}. 
  Therefore $ -*\bs(0) = \as'(0)$ by saturation
   of the Schwarz   inequality. That is, $h_+(A) =0$. In the case of $h_-$ one finds
    $( *\bs(0), \as'(0)) = \| \as'(0)\|_2^2$ and therefore $ *\bs(0) = \as'(0)$. 
    Hence   $h_-(A) =0$.
  \end{proof}

           \bigskip
           \noindent
    \begin{proof}[Proof of Theorem \ref{thmf1}]   
    Items a) and b) of the theorem are proved in Lemma \ref{lemf2}.

   Heuristic argument for \eref{py160}:  The existence of 
\beq
A_+ := \lim_{t\to \infty} (\exp th_+)A
\eeq
 in the $L^2(\R^3)$ sense of convergence is reasonable because $\P$
is decreasing on the orbit of the flow by Lemma \ref{lemf5}, 
while $\{A: \P(A) \le const.\}$ (modulo gauge transformations) is compact in $L^2$ norm
 over each bounded set in $\R^3$, since $\P$ has nominally a Sobolev $H_{1/2}$ strength
 (modulo gauge transformations).
Assuming then the existence of the limit,  $\P(A_+)$ cannot decrease anymore under the flow.
 Hence $h_+(A_+)\P(A_+) =0$.
Therefore $h_+(A_+) =0$ by Lemma \ref{lemf6}. So $A_+ \in \A_+$ by \eref{f13}.  
Conversely, if $A \in \A_+$ then $h_+(A) =0$ by \eref{f13}. 
Therefore $\exp(th_+) A = A$. So $A_+ = A$ if $A \in \A_+$. A similar argument applies to $A_-$.
This proves item c) of the theorem.
 Of course this ensures that the maps $P_\pm: \A \rightarrow \A_{\pm}$  are surjective.

For the proof of \eref{py163} suppose that $A\in \A_+$ and that $u$ is tangential to $\A_+$ at $A$.
Since $h_+$ is identically zero on $\A_+$ its tangential derivative $\p_u h_+(A)$ is zero.
 That is, by \eref{pe150},
\begin{align}
 \ua'(0) + *d_Au =0,
\end{align}
where $\ua(s)$ is the Poisson extension of $u$ along $\as$.
 So $curl_A u = *d_Au=-\ua'(0)$. From the second derivative formula \eref{py209} we find
 $\p_u^2\P(A) = -2(u, \ua'(0))$ for all $A$ and $u$. In particular 
 \beq
 \p_u^2\P(A)  =2(u,curl_A u)\ \ \ \text{if}\ \ A \in \A_+\ \ \text{and $u$ is tangential to}\ \A_+. \label{py170}
 \eeq
 From our convexity assumption on the Poisson action at $A$ it now follows that \eref{py163} holds.
In case $A \in\A_-$ we find  $curl_A u = \ua'(0)$ and therefore \eref{py164} holds.
\end{proof}

\begin{remark}{\rm 
The intersection $\A_+\cap \A_-$ consists of pure gauges. Indeed if $A$ lies
 in this intersection then $h_+(A) = h_-(A) = 0$. Therefore $*B =0$ and $\as'(0) = 0$. Since $B =0$,
 $A$ is a pure gauge.
 }
 \end{remark}

\subsection{The flows in the electromagnetic case}       \label{secfem} 
In the electromagnetic case the vector fields $h_\pm$ are given  by
\begin{align}
    h_\pm (A) = (-|C| \pm C) A,      \label{f104}
    \end{align}
as we see from the argument showing equality between \eref{py140} and \eref{py141}.
We can operate in Coulomb gauge in this example. The limits of the flows of $h_\pm$ are given in the following theorem.

\begin{theorem} \label{thmfem} $($Electromagnetic case.$)$ Denote by $P_\pm$ the orthogonal projections of $\cg$ onto $\cg_\pm$ respectively. Then 
\begin{align}
\lim_{t\to \infty} \exp(th_\pm) A = P_\pm A\ \ \ \text{for}\ \ A \in \cg.    \label{f103}
\end{align}
\end{theorem}
    \begin{proof}    The flow equations  for these vector fields are  the linear differential equations
 \begin{align}
(d/dt) A(t) = (- |C| \pm C) A(t),           \notag
\end{align}  
whose solutions with initial value $A$ are  correctly given by 
\begin{align}
A(t) = e^{ t (-|C| \pm C)} A, \ \ \ t \ge 0            \notag
\end{align}
because $-|C| \pm C \le 0$ in both cases. But $-|C| +C$ is zero on $\cg_+$ and is $-2|C|$ on $\cg_-$.
Therefore, writing $A = P_+A + P_-A$, we find
\begin{align}
A(t) =   e^{ t (-|C| + C)}(P_+A + P_-A) = P_+A  + e^{-2t|C|}\, P_- A,     \notag
\end{align}
which converges in $\cg$ norm to $P_+A$ as $t \to \infty$ because  $|C| \ge 0$ and has trivial null space
in our Coulomb gauge. This proves \eref{f103} in the case of $h_+$.    The case of $h_-$ is similar
because $-|C| - C$ is zero on $\cg_-$ and is $-2|C|$ on $\cg_+$.
\end{proof}

\bigskip

\section{Open questions} \label{secoq}
\

1. Existence and uniqueness of solutions to the Euler equation \eref{pe6} for the minimization of
     the Poisson action have been assumed throughout the sections on the Yang-Mills helicity theory.
     A slight variant of this non-linear boundary value problem has been investigated  in deep work by 
     Marini and Isobe cited in Section  \ref{secpacs}. The form of the existence and uniqueness theorem 
     we need is still open.

     2. The soft Coulomb gauge and asymptotic Coulomb gauge were defined in Section \ref{secpacs}
    under the assumption that $\lim_{s\to\infty} \as(s)$ exists in some useful sense. It needs to be proven
    that this limit actually exists for $H_{1/2}$ initial data  for the Yang-Mills-Poisson equation 
    and is a pure gauge. 
    Such a limit theorem has been established for the Yang-Mills heat equation in various cases.
     See Remark \ref{remsoft}.

3. The flows of the vector fields $h_\pm$ have been assumed to exist in Theorem \ref{thmf1}
 for all positive time. 
This needs to  be proven along with the existence of the limits as time goes to infinity.  The resulting
limit maps $P_\pm$ are gauge invariant and therefore produce a map $P_+ \times P_-: \cg\rightarrow
\cg_+ \times \cg_-$. For electromagnetism this gives the linear decomposition $\cg = \cg_+\oplus \cg_-$.
In the non-linear case it remains to be seen whether this map is one-to-one.

 4. The region of convexity of the Poisson action function $\P$ should be understood better.
 Is there a neighborhood of $A=0$ on which $\P$ is convex? Is there a  boundary to this region
 which somehow reflects the Gribov ambiguity?

 5. The Lorentz invariance of the natural metric on phase space should be proven or disproven
  in the non-abelian theory, 
  as already pointed out in Remark \ref{rmkLinvce}. The first step would be to prove invariance
  under time translation, just  at an informal level for, say, smooth initial data. If such invariance should hold
  then the metric itself would be a  candidate for a useful  invariant for proving short and long time
   existence of solutions to the hyperbolic Yang-Mills equation with 
   $H_{1/2} \oplus H_{-1/2}$ initial data. 
   Existence of  solutions to this equation has not yet been proven for these critical initial data over 
   three dimensional space. See e.g. \cite{Seg111,EM1,EM2,KM,Tao,Oh1,Oh2,ST}.            
   Invariance under Lorentz boosts requires gauge transforming a boosted solution to temporal gauge
   as part of the boost transformation.

\section{Bibliography}

\bibliographystyle{amsplain}
\bibliography{ymh}

\end{document}